\documentclass[11pt, leqno, letterpaper]{amsart}

\usepackage{graphicx}    

\usepackage[margin=1.2in]{geometry}

\usepackage{amsmath,amsthm,amsfonts,amssymb}

\newtheorem{theorem}{Theorem}

\newtheorem{definition}[theorem]{Definition}

\newtheorem{lemma}[theorem]{Lemma}

\newtheorem{proposition}[theorem]{Proposition}

\newtheorem{remark}[theorem]{Remark}

\newtheorem{assumption}[theorem]{Assumption}

\begin{document}

\title[Euler-Discretized Hull-White Stochastic Volatility Model]
{Asymptotics for the Euler-Discretized Hull-White Stochastic Volatility Model}

\author{Dan Pirjol}
\email{
dpirjol@gmail.com}

\author{Lingjiong Zhu}
\address
{Department of Mathematics\newline
\indent Florida State University\newline
\indent 1017 Academic Way\newline
\indent Tallahassee, FL-32306\newline
\indent United States of America}
\email{
ling@cims.nyu.edu}

\date{27 April 2016. \textit{Final:} 4 July 2017}


\begin{abstract}
We consider the stochastic volatility model $dS_t = \sigma_t S_t dW_t,
d\sigma_t = \omega \sigma_t dZ_t$, with $(W_t,Z_t)$ uncorrelated
standard Brownian motions. This is a special case of the Hull-White and the
$\beta=1$ (log-normal) SABR model, which are widely used in financial practice. 
We study the properties of this model, discretized in time under several
applications of the Euler-Maruyama scheme, and point out that the resulting 
model has certain properties which are different from those of the continuous 
time model. 
We study the asymptotics of the time-discretized model in the $n\to \infty$ limit
of a very large number of time steps of size $\tau$, at fixed $\beta=\frac12\omega^2\tau n^2$
and $\rho=\sigma_0^2\tau$, and derive three results: i) almost sure limits, 
ii) fluctuation results, and iii) explicit expressions for growth rates 
(Lyapunov exponents) of the positive integer moments of $S_t$. 
Under the Euler-Maruyama discretization for $(S_t,\log \sigma_t)$, the Lyapunov 
exponents have a phase transition, which appears in numerical simulations
of the model as a numerical explosion of the asset price moments.
We derive criteria for the appearance of these explosions.
\end{abstract}

\maketitle

\section{Introduction} 

Stochastic volatility models are widely used in financial practice for modeling
the dynamics of the volatility surface. Some of the most popular models are
affine models with stochastic volatility such as the Heston model, 
and models where the volatility is the exponential of a Gaussian process 
such as a Brownian motion or Ornstein-Uhlenbeck process.

In this paper we will study the stochastic volatility model defined by the 
process
\begin{eqnarray}\label{SDES}
&& dS_t = \sigma_t S_t dW_t \\
\label{SDEsigma}
&& d\sigma_t = \omega\sigma_t dZ_t
\end{eqnarray}
with initial condition $\sigma_0, S_0$,
where $W_t$ and $Z_t$ denote independent standard Brownian
motions. The model parameter $\omega>0$ is a positive real constant. 

This model is a particular case of the Hull-White model \cite{HW}
\begin{eqnarray}
&& dS_t = r S_t dt + \sqrt{Y_t} S_t 
(\varrho dW_t + \sqrt{1-\varrho^2} dB_t ) \\
&& dY_t = \mu Y_t dt + \zeta Y_t dW_t
\end{eqnarray}
where $W_t,B_t$ are uncorrelated standard Brownian motions. This reduces to
the model (\ref{SDES}),(\ref{SDEsigma}) by identifying $\sigma_t  =\sqrt{Y_t}$ 
and taking $\mu=\frac14 \zeta^2, \zeta = 2\omega, r = 0$ and $\varrho=0$.

The case of zero correlation $\varrho = 0$ has received special 
attention in the literature because of its analytical tractability \cite{GS,WPW}.
This is also a particular realization of the SABR model \cite{SABR}, 
corresponding to the so-called log-normal SABR model $(\beta=1)$
\begin{eqnarray}\label{SABRdef}
&& dS_t = \sigma_t S_t^\beta dW_t \\
&& d\sigma_t = \omega \sigma_t dZ_t \,, \quad 
\mathbb{E}[dW_t,dZ_t]  = \varrho dt\,.
\end{eqnarray} 
The model is also a limiting case (zero mean reversion) of the
Scott model \cite{Sco87,CS89}, which corresponds to assuming that
$\sigma_t$ is the exponential of an Ornstein-Uhlenbeck process. 

In practice the original SABR model as formulated in \cite{SABR}
is used mostly in parametric form for interpolating swaption or caplet
volatilities, due to the unrealistic assumption of zero mean reversion for the
volatility process. Nevertheless, due to its simplicity the model was studied
extensively and many analytical results are available. 
The properties of this model were studied in continuous time in 
\cite{Jourdain,AP2007,LionsMusiela}. We briefly summarize a few results.

The asset price $S_t$ is a strict martingale only if the correlation 
is non-positive $\varrho \leq 0$ \cite{Jourdain,LionsMusiela}, see also
\cite{Sin,BCM}.
This is a necessary condition if the diffusion (\ref{SDES}) is to be used to
model the price of a tradeable asset. 
The model (\ref{SDES}), (\ref{SDEsigma}) has also moment explosions. 
Moment explosions in stochastic volatility models have been studied
in a wide class of models, see \cite{AP2007,GK,review}.
Define the explosion time $T_*(q)$ of the $q$-th moment of 
the asset price $(S_t)^q$ with $q\in \mathbb{R}$ as 
\begin{equation}
T_*(q) = \mbox{sup} \{ t\geq 0 : \mathbb{E}[S_t^q] < \infty \} 
\end{equation}
For the model (\ref{SDES}), (\ref{SDEsigma}) the explosion time is given 
by \cite{Jourdain,LionsMusiela}
\begin{equation}\label{Tstaru}
T_*(q) = 
\begin{cases}
+\infty\,, & \varrho < \varrho^*(q) \\
0\,,       & \varrho > \varrho^*(q) 
\end{cases}
\end{equation}
where the critical correlation $\varrho^*(q)$ is 
\begin{equation}
\varrho^*(q) = - \sqrt{1-\frac{1}{q}} \,.
\end{equation}
For all positive moments $q>1$ the critical correlation is negative 
$\varrho^*(q) < 0$, such that at
zero correlation $\varrho = 0$ all these moments explode in zero time. 

In practical implementation using Monte Carlo approaches,  
stochastic volatility models are simulated in discrete time. 
For this purpose, the stochastic differential equation (\ref{SDES}),
(\ref{SDEsigma}) is discretized in time, using one of several available time 
discretization schemes \cite{KlPl,Jourdain2}. The simplest scheme is the forward 
Euler time discretization, or Euler-Maruyama discretization \cite{KlPl}. 
This can be applied either directly to the stochastic differential equation
for $S_t$, or to that for $X_t = \log S_t$. We will call these schemes the 
Euler and log-Euler
schemes, respectively. The same treatment can be applied to the stochastic
differential equation for $\sigma_t$. Although the latter can be solved exactly
for this case, we consider also its discretization as an illustration for
more complicated volatility processes, where an exact treatment is not available.

The Euler-Maruyama time discretizations of the stochastic volatility model 
(\ref{SDES}), (\ref{SDEsigma}) have distinctive
properties which can be different from those of the continuous-time model. 
For example, the asset price $S_n$ is a true martingale for any correlation 
$\varrho \in [-1,1]$ (Proposition~\ref{prop:martingale}), in contrast to the 
continuous time model where this property holds only for $\varrho \leq 0$.
Also, under Euler-Maruyama discretization of $(S_t, \log\sigma_t)$, 
the moments of the asset price $\mathbb{E}[(S_n)^q]$ are finite for any 
parameter values $(\sigma_0,\omega,\tau)$. On the other hand, in continuous 
time, as noted above, all moments $\mathbb{E}[(S_t)^q]$ with $q>1$ explode in 
zero time \cite{Jourdain,LionsMusiela}.

One surprising feature of the Euler-discretized model is that the
positive integer moments $\mathbb{E}[(S_n)^q]$ with $q\geq 2$ 
have a sudden rapid increase for sufficiently large $\omega$ or simulation
time step $n$. This phenomenon is well known to practitioners, and
is known to appear in simulations of stochastic volatility models with 
log-normally distributed volatility. See \cite{JK} for an informal discussion.
For a discussion in the context of Monte Carlo simulations of the Hull-White
model with arbitrary correlation, see Sec.~5.2 in \cite{GHS}, where 
the explosion of the variance of the asset price is controlled by imposing an 
upper bound on the values of the stochastic volatility process $\sigma_t$.

This phenomenon introduces difficulties in the estimation of the 
the error of Monte Carlo pricing of payoffs $f(S_t)$, since a very large 
variance of the asset price $S_t$ may lead to a very large variance of the
payoff $f(S_t)$ \cite{PGbook}. Large values of the higher moments can have 
also direct relevance for pricing certain instruments. For example, in fixed 
income markets, the second moment of forward Libor rates is relevant
for pricing certain instruments such as Libor payments in arrears \cite{AP2007}.

The explosion of moments observed for the time discretization of the model
(\ref{SDES}), (\ref{SDEsigma}) appears in a wider class of models. It was
observed in a discrete time stochastic compounding process 
$x_{i+1} = x_i(1 + \rho e^{\sigma W_i - \frac12\sigma^2 t_i})$
with multipliers
proportional to a geometric Brownian motion \cite{RMP}. The positive integer
moments of the compounding process $\mathbb{E}[(x_n)^q]$ were observed to 
explode for sufficiently large values of the volatility $\sigma$ or time step 
$n$. We emphasize that the explosion is to very
large (but finite) values, and the moments remain strictly finite, as expected
in a discrete time setting. 
This phenomenon was studied in \cite{LD}
using large deviations theory, where it was shown that the moment explosion
is due to a discontinuous behavior of the Lyapunov exponents $\lambda_q = 
\lim_{n\to \infty} \frac{1}{n} \log\mathbb{E}[(x_n)^q]$. The existence of this
limit requires that the model parameters are rescaled with $n$ such that a 
certain combination $\beta = \frac12\sigma^2 \tau n^2$ is kept finite in the 
$n\to \infty$ limit. We will use in this paper a similar approach to study the
phenomenon of moment explosion in the Euler discretized version of the
stochastic volatility model (\ref{SDES}), (\ref{SDEsigma}).

We study in this paper the large $n$ asymptotics of $S_n$ as $n\to \infty$
in the model (\ref{SDES}), (\ref{SDEsigma}) discretized in time under 
various applications of the Euler-Maruyama scheme. In usual applications
of the Euler discretization one is interested in the $n\to \infty$ limit
of a very large number of time steps at fixed maturity $n\tau = T$. 
The limit considered here is different, as we take $n\to \infty$ at fixed
$\beta = \frac12 \omega^2 \tau n^2$ and $\rho = \sigma_0^2\tau$.
As explained in the next section, this covers the fixed-maturity limit
$n\tau = $ fixed, with small $\omega \sim O(n^{-1/2})$ and large 
$\sigma_0 \sim O(n^{1/2})$ (denoted as Regime 3 below). In addition, this
scaling includes other regimes, corresponding to large maturity 
$n\tau \sim O(n)$, small $\omega \sim O(n^{-1})$ (Regime 1), and small maturity 
$n\tau \sim O(n^{-1})$, large $\sigma_0 \sim O(n)$ (Regime 2). 

As mentioned, the Euler-Maruyama time discretized versions of the
stochastic volatility model (\ref{SDES}) can have different properties
from those of the continuous time model. The different $n\to \infty$ limit
considered here introduces also differences in the asymptotics of the
discrete time model compared to those of the usual Euler-Maruyama discretized
model. For example the limits $\lim_{n\to \infty} S_n$ may be different from
the corresponding large time limit of the continuous time model. 

The motivation for adopting the specific large $n$ scaling of this paper
is that the growth rates (Lyapunov exponents) of the positive integer moments 
$\lambda_q = \lim_{n\to \infty} \frac{1}{n} \log \mathbb{E}[(S_n)^q]$ 
exist and are finite under this scaling. 
We obtain explicit results for the Lyapunov exponents 
under this special scaling of the model parameters using large deviations 
theory, and study their functional dependence on the model parameters. 
We find that, under the application of the Euler-Maruyama 
scheme to $(S_t, \log\sigma_t)$, the Lyapunov exponents have non-analyticity 
in the model parameters which is similar to the phase transition studied in 
\cite{LD}. This phenomenon is responsible for the numerical explosions of the
moments of the asset price observed in numerical simulations of the model, 
which have implications for the Monte Carlo simulation of the model as discussed
above.

Section 2 introduces the different Euler-Maruyama schemes, and the scaling 
of the model parameters under the large $n$ limit considered in this paper. 
The asymptotic properties of the different schemes
are discussed separately: the Euler-Log Euler scheme (Sec.~3), the
Log Euler-Log Euler scheme (Sec.~4), Log Euler-Euler scheme (Sec.~5) and
Euler-Euler scheme (Sec.~6). Section 7 presents a detailed comparison
of the asymptotics of these schemes with the known results of the continuous
time model, and we demonstrate good agreement between the properties of the
phase transition for Euler-Log Euler scheme obtained from the asymptotic
analysis of Section 4 with exact numerical simulations of the model. This 
agreement demonstrates the practical usefulness of our asymptotic results,
as they give thresholds for the numerical explosion of the moments observed
in numerical simulations of the model.

\section{Euler-Maruyama Time Discretizations}
\label{sec:2}

Stochastic volatility models are usually simulated in practice in 
discrete time. Finite grid simulations discretize the time, asset, and 
volatility, while Monte Carlo simulation discretize only in time.

Consider the simulation of the one-dimensional It\^o stochastic differential 
equation for the $k$-dimensional stochastic vector $X_t$ for $t\in [0,T]$
\begin{equation}\label{SDE}
dX_t = \sigma(X_t,t) dW_t + b(X_t,t) dt
\end{equation}
with initial condition $X_0 = x_0$ and coefficients 
$\sigma(z,t), b(z,t):\mathbb{R} \times [0,T] \to \mathbb{R}$ which are
globally Lipschitz functions. This ensures the existence
of strong solutions \cite{KlPl,Lewis}. $W_t$ is a standard Brownian motion. 

We divide the interval $[0,T]$ into $n$ time steps with uniform length 
$\tau = T/n$. We would like to simulate the SDE (\ref{SDE}) on the sequence of
discrete time steps
\begin{equation}
0 = t_0 < t_1 < t_2 < \cdots < t_n = T \,.
\end{equation}
The simplest time discretization is the explicit Euler scheme, or
Euler-Maruyama discretization, which is defined by the recursion
\begin{equation}
X_{i+1} = X_i + \sigma(X_i, t_i) (W_{i+1} - W_i) + b(X_i, t_i) (t_{i+1}-t_i)
\end{equation}
with initial condition $X_0 = x_0$. 
The convergence properties of the Euler-Maruyama scheme were studied in
\cite{TT,BT1,Guyon}.

We can apply the Euler-Maruyama discretization either to $S_t$ or to the
log-price $\log S_t$, and same for $\sigma_t$.
Upon taking all possible combinations, we have altogether four 
possible discretizations: Euler-log-Euler scheme, 
Log-Euler-log-Euler scheme, Euler-Euler scheme and Log-Euler-Euler scheme. 
For instance in Euler-log-Euler scheme, the first ``Euler''
refers to discretization of the asset price $S_{t}$ and the second 
``log-Euler'' refers to the discretization of the volatility $\sigma_{t}$
and so on and so forth. 

\begin{definition}[Euler-log-Euler scheme]
The Euler-Maruyama scheme (or simply Euler-log-Euler scheme) for the discretization of 
the SDE (\ref{SDES}) is defined by
\begin{align}\label{EulerS}
&S_{i+1}= S_{i}\left(1+\sigma_{i}\sqrt{\tau}\varepsilon_{i}\right) 
\\
&\sigma_{i+1}=\sigma_{i}\exp\left(\omega(Z_{i+1}-Z_{i})-\frac{1}{2}\omega^{2}\tau\right).
\nonumber
\end{align}
\end{definition}

We denote for simplicity $\sigma_i = \sigma(t_i)$, and
$\Delta W_i = W_{i+1} - W_i = \sqrt{\tau} \varepsilon_i$ 
and $Z_{i+1}-Z_{i}=\sqrt{\tau}V_{i}$, 
with
$\varepsilon_k \sim N(0,1)$ i.i.d. Gaussian variables with mean $0$
and variance $1$
independent of i.i.d. Gaussian variables $V_{i}\sim N(0,1)$ with mean $0$ and variance $1$.

The log-Euler discretization for $\sigma_{t}$ coincides with the exact solution for the volatility at the 
start of the $(t_i, t_{i+1})$ time interval since $\omega$ is a constant
\begin{equation}
\sigma_{i}=\sigma_{0}e^{\omega Z_{i}-\frac{1}{2}\omega^{2}t_{i}}.
\end{equation}

This scheme has the disadvantage that the asset price $S_n$ can become negative.
An alternative scheme which preserves the positivity of the asset price 
is defined by applying the Euler scheme to $\log S_n$. This is given by
the following recursion:

\begin{definition}[Log-Euler-log-Euler scheme]
The Euler-Maruyama scheme for $\log S_n$ (or 
the log-Euler-log-Euler scheme) for the discretization of 
the SDE (\ref{SDES}) is defined by
\begin{align}\label{LogEulerS}
&S_{i+1}=S_{i} 
\exp\left(\sigma_i (W_{i+1} - W_i) - \frac12 \sigma_i^2 \tau \right)
\\
&\sigma_{i+1}=\sigma_{i}\exp\left(\omega(Z_{i+1}-Z_{i})-\frac{1}{2}\omega^{2}\tau\right).
\nonumber
\end{align}
\end{definition}

\begin{definition}[Euler-Euler scheme]
\begin{align}\label{EulerEulerS}
&S_{i+1}= S_{i}\left(1+\sigma_{i}\sqrt{\tau}\varepsilon_{i}\right) 
\\
&\sigma_{i+1}=\sigma_{i}\left(1+\omega\sqrt{\tau}V_{i}\right).
\nonumber
\end{align}
\end{definition}

\begin{definition}[Log-Euler-Euler scheme]
\begin{align}\label{LogEulerEulerS}
&S_{i+1}=S_{i} 
\exp\left(\sigma_{i}(W_{i+1}-W_i)-\frac{1}{2}\sigma_{i}^{2}\tau\right)
\\
&\sigma_{i+1}=\sigma_{i}\left(1+\omega\sqrt{\tau}V_{i}\right).
\nonumber
\end{align}
\end{definition}

Throughout the paper, we assume the following:
\begin{assumption}\label{assump5}
We assume that
\begin{align}
&\beta:=\frac{1}{2}\omega^{2}n^{2}\tau,\label{AssumpI}
\\
&\rho:=\sigma_{0}\sqrt{\tau},\label{AssumpII}
\end{align}
are fixed positive constants\footnote{Note the change of notation for 
$\beta$ from the equation (\ref{SABRdef}). In the remainder of the paper
$\beta$ will be defined as in (\ref{AssumpI}).}.
\end{assumption}
The model parameters $\sigma_0,\omega,\tau$ may depend on $n$.
Our results should hold also under the less restrictive conditions
$\beta =\lim_{n\to \infty} \frac{1}{2}\omega^{2}n^{2}\tau$,
$\rho = \lim_{n\to \infty}\sigma_{0}\sqrt{\tau}$, but for the sake
of simplicity we will use the definitions in Assumption~\ref{assump5}.

In usual applications of the Euler-Maruyama discretization 
the model parameters $\sigma_0,\omega$ are kept fixed as $\tau \to 0$.
However, as shown above, in the continuous time case, the moments
$\mathbb{E}[(S_T)^q]$ will be infinite for any $q>1$. On the other hand,
under the scalings (\ref{AssumpI}), (\ref{AssumpII}) the Lyapunov exponents
$\lambda_q$ of these moments will be seen to be finite. (The moments themselves 
will approach infinity as $n\to \infty$ since 
$\mathbb{E}[(S_n)^q] = e^{\lambda_q n + o(n)}$.)
Thus the scalings (\ref{AssumpI}), (\ref{AssumpII}) ensure that a well-defined
limit exists for the growth rates of the moments (Lyapunov exponents). 

Our assumptions \eqref{AssumpI}, \eqref{AssumpII} include the following 
asymptotic regimes:
\begin{itemize}
\item
Regime 1. When $\tau$ and $\sigma_{0}$ are fixed constants, the volatility of 
volatility $\omega$ is of order $O(\frac{1}{n})$ and the maturity $t_{n}=n\tau$ 
is of order $O(n)$. So this is the small volatility of volatility and large 
maturity regime.

\item
Regime 2. When $\tau$ is of the order $O(\frac{1}{n^{2}})$, then $t_n \sim O(n^{-1})$
such  that $\sigma_{0}$ is of the order $O(n)$ and $\omega \sim O(1)$. 
This corresponds to the large initial volatility and small maturity regime.

\item
Regime 3. When $\tau$ is of the order $O(\frac{1}{n})$, then $\sigma_{0}$ is of the 
order $O(n^{\frac12})$, and $\omega$ is of the order $O(n^{-\frac12})$.
This corresponds to the fixed maturity regime.

\end{itemize}

The accuracy of the asymptotic results for the Lyapunov exponent increases
with the number of steps $n$.
We illustrate the application of the asymptotic 
results with numerical examples in Section~\ref{sec:7}, where we show that
the asymptotic limit gives a reasonably good approximation for the finite $n$
result of the Lyapunov exponents $\lambda_{q,n}  = \frac{1}{n}\log \mathbb{E}[
(S_n)^q]$ for values of $n$ as low as 40.

\section{Euler-Log-Euler Scheme}
\label{EulerLogEulerSection}

\subsection{Lyapunov Exponents of the Moments}

We would like to compute the moments of the asset price $S_n$
in the time discretizations introduced in the previous section. We consider
in this section the scheme (\ref{EulerS}).
Let us recall that this scheme is defined by the recursion
\begin{align}
&S_{i+1}= S_{i}\left(1+\sigma_{i}\sqrt{\tau}\varepsilon_{i}\right) 
\\
&\sigma_{i+1}=\sigma_{i}\exp\left(\omega(Z_{i+1}-Z_{i})-\frac{1}{2}\omega^{2}\tau\right).
\nonumber
\end{align}
with initial condition $S_0, \sigma_0$.

The $q$th moments of $S_n$, $q\in\mathbb{N}$ are given by
\begin{align}
\mathbb{E}[(S_{n})^{q}]
&=S_{0}^{q}\mathbb{E}\left[\prod_{k=0}^{n-1}\left(1+\sigma_{0}\sqrt{\tau}
e^{\omega Z_{k}-\frac{1}{2}\omega^{2}t_{k}}\varepsilon_{k}\right)^{q}\right]
\\
&=S_{0}^{q}\mathbb{E}\left[\prod_{k=0}^{n-1}\sum_{j=0}^{[q/2]}\frac{q!}{(2j)!(q-2j)!}(2j-1)!!(\sigma_{0}\sqrt{\tau})^{2j}
e^{2j(\omega Z_{k}-\frac{1}{2}\omega^{2}t_{k})}\right],
\nonumber
\end{align}
where we took the expectations over $(\varepsilon_{k})_{0\leq k\leq n-1}$ and used the identity for the $q$th moments
of Gaussian random variables
\begin{equation}
\mathbb{E}\left[\left(1+\sigma\varepsilon_{1}\right)^{q}\right]
=\sum_{j=0}^{[q/2]}\frac{q!}{(2j)!(q-2j)!}(2j-1)!!\sigma^{2j}.
\end{equation}
Since $\rho=\sigma_{0}\sqrt{\tau}$ and
$\beta=\frac{1}{2}\omega^{2}n^{2}\tau$,
we can express the $q$th moments by
\begin{equation}
\mathbb{E}[(S_{n})^{q}] 
=S_{0}^{q}m(q)^{n}\mathbb{E}\left[\prod_{k=0}^{n-1}e^{\log\rho Y_{k}+\omega Z_{k}Y_{k}-\frac{1}{2}\omega^{2}t_{k}Y_{k}}\right],
\nonumber
\end{equation}
where
$(Y_{k})_{k=0}^{n-1}$ are i.i.d. random variables such that
\begin{equation}
\mathbb{P}(Y_{k}=2j)=\frac{1}{m(q)}\frac{q!}{(2j)!(q-2j)!}(2j-1)!!,
\qquad
j=0,1,2,\ldots,[q/2].
\end{equation}
$m(q)$ is the normalizer defined as
\begin{equation}
m(q)=\sum_{j=0}^{[q/2]}\frac{q!}{(2j)!(q-2j)!}(2j-1)!!.
\end{equation}

Now, observe that $0\leq Y_{k}\leq q$ and $0\leq\omega^{2}t_{k}\leq\omega^{2}t_{n}=\omega^{2}\tau n=\frac{2\beta}{n}$.
Therefore,
\begin{align}
\mathbb{E}\left[\prod_{k=0}^{n-1}e^{\log\rho Y_{k}+\omega Z_{k}Y_{k}}\right]e^{-\beta}
&\leq
\mathbb{E}\left[\prod_{k=0}^{n-1}e^{\log\rho Y_{k}+\omega Z_{k}Y_{k}-\frac{1}{2}\omega^{2}t_{k}Y_{k}}\right]
\\
&\leq\mathbb{E}\left[\prod_{k=0}^{n-1}e^{\log\rho Y_{k}+\omega Z_{k}Y_{k}}\right].
\nonumber
\end{align}
Hence
\begin{equation}
\lim_{n\rightarrow\infty}\frac{1}{n}\log\mathbb{E}[(S_{n})^{q}] 
=\log m(q)+\lim_{n\rightarrow\infty}\frac{1}{n}\log\mathbb{E}\left[\prod_{k=0}^{n-1}e^{\log\rho Y_{k}+\omega Z_{k}Y_{k}}\right],
\end{equation}
if the limit exists. Indeed, using Large Deviations theory we will show that 
this limit exists, and is given by the following result.

\begin{theorem}\label{EulerLogEulerThm}
For any $q\in\mathbb{N}$, $\lambda(\rho,\beta;q):=\lim_{n\rightarrow\infty}\frac{1}{n}\log\mathbb{E}[(S_{n})^{q}]$
exists and it can be expressed in terms of a variational formula
\begin{equation}
\lambda(\rho,\beta;q)=\sup_{g\in\mathcal{G}_{q}}
\left\{g(1)\log\rho+\beta\int_{0}^{1}(g(1)-g(x))^{2}dx-\int_{0}^{1}I_{q}(g'(x))dx\right\},
\end{equation}
where $I_{q}(x):=\sup_{\theta\in\mathbb{R}}\{\theta x-f_{q}(\theta)\}$, where
\begin{equation}\label{IqFormula}
f_{q}(\theta)
:=\log\mathbb{E}\left[\left(1+\varepsilon_{1}e^{\theta}\right)^{q}\right]=\log\sum_{j=0}^{[q/2]}\frac{q!}{(2j)!(q-2j)!}(2j-1)!!e^{2j\theta},
\end{equation}
where $\varepsilon_{1}$ is a normal random variable with mean $0$ and variance $1$ and
\begin{align}\label{GqFormula}
\mathcal{G}_{q}&:=\bigg\{g:[0,1]\rightarrow[0,2[q/2]], 
\text{$g(0)=0$, $g$ is absolutely continuous} 
\\
&\qquad\qquad\qquad\qquad\qquad
\text{and $0\leq g'\leq 2[q/2]$}\bigg\}.
\nonumber
\end{align}
\end{theorem}

\begin{proof}
The proof will be given in Section \ref{AppendixSection}.
\end{proof}

\begin{remark}
We have studied the limit $\lim_{n\rightarrow\infty}\frac{1}{n}\log\mathbb{E}[(S_{n})^{q}]$ for any non-negative integer $q$.
Observe that $S_{n}$ can take negative values, therefore $(S_{n})^{q}$ is not well-defined for non-integer $q$.
If $q$ is a negative integer, $\mathbb{E}[X^{q}]$ is not well defined for any Gaussian random variable $X$
since $\frac{1}{|x|^{|q|}}e^{-\frac{(x-\mu)^{2}}{2\sigma^{2}}}$ is not Lebesgue integrable for any open interval including $x=0$.
Therefore, we restrict ourselves to the study of non-negative integer moments.
\end{remark}

\begin{remark}
Recall that
\begin{equation}
\mathbb{E}\left[\left(1+\varepsilon_{1}e^{\theta}\right)^{q}\right]
=\sum_{j=0}^{[q/2]}\frac{q!}{(2j)!(q-2j)!}(2j-1)!!e^{2j\theta},
\end{equation}
which equals to the $q$th moment of a normal random variable with mean $1$ and variance $e^{2\theta}$.
It is known that
\begin{equation}
\mathbb{E}[N(\mu,\sigma^{2})^{q}]
=(-2\sigma^{2})^{\frac{q}{2}}U\left(\frac{-q}{2},\frac{1}{2},\frac{-\mu^{2}}{2\sigma^{2}}\right),
\end{equation}
where $U$ is a confluent hypergeometric function and the function's second branch cut
can be chosen by multiplying with $(-1)^{q}$. For general $q\geq 2$, it seems to be difficult
to get a more explicit expression. But it is possible at least for $q=2,3,4,5$.

(i) For $q=2$, 
\begin{equation}
I_{2}(x)=\sup_{\theta\in\mathbb{R}}\left\{\theta x-\log(1+e^{2\theta})\right\}=
\frac12 x\log x + \frac12 (2-x)\log(2-x) - \log 2\,.
\end{equation}

(ii) For $q=3$,
\begin{equation}
I_{3}(x)=\sup_{\theta\in\mathbb{R}}\left\{\theta x-\log(1+3e^{2\theta})\right\}
=\frac{x}{2}\log\left(\frac{x}{6-3x}\right)-\log\left(\frac{6}{6-3x}\right).
\end{equation}

(iii) For $q=4$,
\begin{align}
I_{4}(x)&=\sup_{\theta\in\mathbb{R}}\left\{\theta x-\log\left(1+\binom{4}{2}\cdot 1\cdot e^{2\theta}
+\binom{4}{4}\cdot 3\cdot 1\cdot e^{4\theta}\right)\right\}
\\
&=\sup_{\theta\in\mathbb{R}}\left\{\theta x-\log\left(1+6e^{2\theta}+3e^{4\theta}\right)\right\}
\nonumber
\\
&=\frac{x}{2}\log \eta_{4}(x)
-\log\left(1+6\eta_{4}(x)+3(\eta_{4}(x))^{2}\right),
\nonumber
\end{align}
where $\eta_{4}(x)=\frac{6(x-2)+\sqrt{6^{2}(2-x)^{2}+3\cdot 4x(4-x)}}{2\cdot 3(4-x)}$.

(iv) For $q=5$,
\begin{align}
I_{5}(x)&=\sup_{\theta\in\mathbb{R}}\left\{\theta x-\log\left(1+\binom{5}{2}\cdot 1\cdot e^{2\theta}
+\binom{5}{4}\cdot 3\cdot 1\cdot e^{4\theta}\right)\right\}
\\
&=\sup_{\theta\in\mathbb{R}}\left\{\theta x-\log\left(1+10e^{2\theta}+15e^{4\theta}\right)\right\}
\nonumber
\\
&=\frac{x}{2}\log \eta_{5}(x)
-\log\left(1+10\eta_{5}(x)+15(\eta_{5}(x))^{2}\right),
\nonumber
\end{align}
where $\eta_{5}(x)=\frac{10(x-2)+\sqrt{10^{2}(2-x)^{2}+15\cdot 4x(4-x)}}{2\cdot 15(4-x)}$.
\end{remark}

The formula for $\lambda(\rho,\beta;q)$ is complicated but the limits for large $\beta$, small $\beta$,
and large $\rho$ are more tractable. We have the following result.

\begin{proposition}\label{AsympProp}
(i) The $\beta \to 0$ limit for the Lyapunov exponent is
\begin{equation}
\lambda(\rho,0;q)=\sup_{0\leq x\leq 2[q/2]}\{x\log\rho-I_{q}(x)\}=f_{q}(\log\rho). 
\end{equation}

(ii) The Lyapunov exponent is bounded from above and below as 
\begin{align}\label{bound1}
&\beta\frac{4}{3}[q/2]^{2}+2[q/2]\log\rho-I_{q}(2[q/2])
\\
&\qquad\qquad\qquad
\leq\lambda(\rho,\beta;q)
\leq\frac{4}{3}\beta[q/2]^{2}
+f_{q}(\log\rho).
\nonumber
\end{align}

(iii) These bounds give the asymptotic behavior in the 
large $\beta$ limit
\begin{equation}\label{LargeBeta}
\lim_{\beta\rightarrow\infty}\frac{\lambda(\rho,\beta;q)}{\beta}=\frac{4}{3}[q/2]^{2},
\end{equation}
and in the large $\rho$ limit
\begin{equation}\label{LargeRho}
\lim_{\rho\rightarrow\infty}\left|\lambda(\rho,\beta;q)-\beta\frac{4}{3}[q/2]^{2}-2[q/2]\log\rho+I_{q}(2[q/2])\right|=0.
\end{equation}
\end{proposition}

\begin{proof}
The proof will be given in Section \ref{AppendixSection}.
\end{proof}

\subsection{Variational Problem}

Consider the variational problem appearing in Theorem~\ref{EulerLogEulerThm}. 
This can be expressed equivalently in terms of the functional
\begin{equation}\label{Lambdadef}
 \Lambda[f] \equiv
\log\rho \int_0^1  f(x) dx+ 
\beta \int_0^1  \left(\int_x^1 f(y)dy\right)^2 dx
- \int_0^1 I_q(f(x)) dx
\end{equation}
defined in terms of a function $f:[0,1]\to [0,2[q/2]]$
subject to the constraints
\begin{equation}
0 \leq f(x) \leq 2[q/2]\,.
\end{equation}
The function $f(x)$ is related to the function $g(x)$ appearing in  
Theorem \ref{EulerLogEulerThm} as $f(x) = g'(x)$.
The rate function $I_q(x)$ is given by
\begin{equation}
I_q(x) \equiv \mbox{sup}_{\theta\in \mathbb{R}} \{\theta x - f_q(\theta)\},
\end{equation}
where $f_q(\theta)$ is given by Equation~(\ref{IqFormula}).

The variational problem for $\Lambda[f]$ gives an integral equation
\begin{equation}\label{EulerLagrange}
\frac{\delta \Lambda[f]}{\delta f} = 
\log\rho + 2\beta  \int_0^1 K(y,z) f(z) dz - I'_q(f(y)) = 0\,,
\end{equation}
where the kernel $K(y,z)$ is $K(y,z) = \min\{y,z\}$.
This can be 
transformed into an ordinary differential equation by taking successive
derivatives with respect to $y$. 
The equation (\ref{EulerLagrange}) is written as
\begin{equation}\label{d0}
\log\rho + 2\beta \int_0^y z f(z) dz + 2\beta y \int_y^1 f(z) dz - I'_q(f(y)) = 0\,.
\end{equation}
Take one derivative with respect to $y$
\begin{equation}\label{d1}
2\beta \int_y^1 f(z) dz - \frac{d}{dy}I'_q(f(y)) = 0\,.
\end{equation}
Taking another derivative gives
\begin{equation}\label{LE}
2\beta f(y) +  \frac{d^2}{dy^2}I'_q(f(y)) = 0\,.
\end{equation}
The function $f(y)$ is given by the solution of this ordinary
differential equation with boundary conditions 
\begin{equation}\label{BC}
\log\rho = I'_q(f(0))\,, \qquad f'(1) = 0\,.
\end{equation}
The first condition follows by taking $y=0$ in (\ref{d0}), and  the second
condition is obtained by taking $y=1$ in (\ref{d1}).

We will prove next that the equation (\ref{LE}) can be formulated in 
such a way that it only requires the cumulant function $f_q(x)$ but not its
Legendre-Fenchel transform $I_q(x)$. Introduce a new
unknown function $h(y)$ defined by 
\begin{equation}
h(y) = I'_q(f(y))\,.
\end{equation}
This is inverted as
\begin{equation}\label{fvsh2}
f(y) = f_q(h(y)) \,.
\end{equation}
The proof of this relation follows from the observation that $I'_q(x) = \theta_*(x)$
where $\theta_*(x)$ is the value of $\theta$ which achieves the supremum in the
definition of the rate function $I_q(x)$. 

In conclusion, the equation satisfied by the function $h(y)$ 
is given by the following result.
\begin{proposition}\label{hProp}
The function $h(y)$ satisfies the second order differential equation
\begin{equation}\label{ODEh}
h''(y) = - 2\beta f_q(h(y)) = - V'(h(y))
\end{equation}
where we defined the potential function $V(h)$ as
\begin{equation}
V(h) = 2\beta f_q(h)\,.
\end{equation}
The boundary conditions for the equation (\ref{ODEh}) are
\begin{equation}\label{BCh}
h(0) = \log\rho\,,\qquad h'(1) = 0\,.
\end{equation}
\end{proposition}
\begin{proof}
This follows directly from combining Equations (\ref{LE}) and (\ref{fvsh2}).
\end{proof}

This reduces  the variational problem to that of finding the solution of an
ordinary differential equation.
The equation (\ref{ODEh}) with boundary conditions (\ref{BCh}) is solved
straightforwardly. We start by noting that the quantity
\begin{equation}
E = \frac12 [h'(y)]^2 + V(h(y)) = V(h(1))
\end{equation}
is a constant of motion.
This can be used to express the derivative $h'(y)$ in terms of $h(y)$ as
\begin{equation}
\frac{dy}{dh} = \frac{1}{\sqrt{2(V(h(1)) - V(h(y))}} = 
\frac{1}{2\sqrt{\beta}} \frac{1}{\sqrt{f_q(h(1)) - f_q(h)}}
\end{equation}
so integrating from $h(0)$ to $h(y)$ we obtain
\begin{equation}\label{yvsh}
y = \frac{1}{2\sqrt{\beta}}  \int_{h(0)=\log\rho}^{h(y)}
\frac{dx}{\sqrt{f_q(h(1)) - f_q(x)}} \,.
\end{equation}

Taking $y=1$ in this relation gives an equation for $h(1)$ 
\begin{equation}\label{h1eq0}
F_q(h(1);\rho) = 2\sqrt{\beta}
\end{equation}
where 
\begin{equation}
F_q(a;\rho) = \int_{\log\rho}^a \frac{dx}{\sqrt{f_q(a) - f_q(x)}}\,.
\end{equation}
Once $h(1)$ is known, the function $h(y)$ can be found by substituting 
its value into (\ref{yvsh}) and solving for $h(y)$.

Another important simplification is that $\Lambda[h]$ can 
be expressed only in terms of $h(1)$. This reduces the variational 
problem to that of finding the extremum of a real function of one 
variable. This is given by the following result
\begin{proposition}\label{prop:LambdaSimple}
The functional $\Lambda[h]$ can be expressed as
\begin{equation}\label{Lambdasimple}
\Lambda[f] = f_q(h(1)) - \frac{1}{\sqrt{\beta}} 
\int_{\log\rho}^{h(1)} \sqrt{f_q(h(1)) - f_q(x)} dx\,.
\end{equation}
\end{proposition}

\begin{proof}
The proof is given in Section \ref{AppendixSection}.
\end{proof}

\begin{remark}
The variational problem for the functional $\Lambda[f]$ defined in
(\ref{Lambdadef}) generalizes the variational problem considered in \cite{LD} to 
a wider class of rate functions $I_q(x)$. The problem considered
in \cite{LD} corresponds to $f_q(x) = \log(1+e^x)$. 
The simpler result (\ref{Lambdasimple}) agrees with the result for
the functional in Proposition 7 of \cite{LD}, upon substituting
$f_q(x) = \log(1+e^x)$ into (\ref{Lambdasimple}).
\end{remark}

\subsection{Numerical Study and Analytical Approximations}

The Lyapunov exponent $\lambda(\rho,\beta;q)$ given by the solution of the
variational problem in Theorem~\ref{EulerLogEulerThm} displays the phenomenon 
of phase transition. This is manifested as a discontinuity of the partial
derivatives $\partial\lambda/\partial\rho$ and $\partial\lambda/\partial\beta$
at points $\beta_c(\rho)$ along a curve in the $(\rho,\beta)$ plane.

This phenomenon was studied in \cite{LD} for the case of the moment generating
function $f_q(x) = \log(1+e^x)$. It can appear when the equation (\ref{h1eq0})
has multiple solutions for $h(1)$, and the optimal value of $h(1)$ switches
between two of these solutions. For sufficiently small values of $\rho < \rho_c$
below a critical value $\rho_c$ the equation (\ref{h1eq0}) has multiple solutions,
while for $\rho > \rho_c$ the solution is unique.

The properties of the phase transition have
been studied in \cite{LD} using both numerical and analytical methods. 
The existence of a phase transition for sufficiently
small $\rho$ and its absence for sufficiently large $\rho$ have been proved
rigorously in Section~7 of \cite{LD}. These proofs can be extended immediately
to the case of the $q=2,3$ moments, for which the corresponding moment
generating functions $f_2(x) = \log(1+e^{2x})$ and $f_3(x) = \log(1+3 e^{2x})$
have a similar functional dependence to the function $f(x) = \log(1+e^x)$
considered in \cite{LD}.
For the higher moments $q>3$ the properties of the phase transition can be
studied numerically.

We present in this section a numerical study of the Lyapunov exponent and its
phase transition for the first few positive integer moments $q=2,\ldots,7$. 

\subsubsection{The second moment ($q=2$)}
\label{sec:q2}

The numerical solution of the variational problem can be simplified by taking 
as independent variable $d = \int_0^1 dx f(x) \in (0,2)$ instead of $h(1)$. 
This is related to $h(1)$ as shown in (\ref{I0def}).
Expressed in terms of this variable, the Lyapunov exponent of the $q=2$ moment 
is obtained by taking the supremum of the following functional over $d\in [0,2]$
\begin{align}
\Lambda_2(d) &= \beta d^2 + \log(1+\rho^2)
\\
& 
\qquad\qquad  - \beta d^3 (1+\rho^2)
\int_0^1 \frac{y^2 dy}{1+\rho^2 - e^{\beta d^2 (y^2-1)}}\,. \nonumber
\end{align}

We computed the Lyapunov exponent of the second moment $\lambda(\rho,\beta;2)$
using this relation. We show in Figure \ref{Fig:lambda2} plots of $\lambda(\rho,\beta;2)$
vs $\beta$ for several values of $\rho$. These results show that for sufficiently
small $\rho < \rho_c^{(2)}$, below a critical value $\rho_c^{(2)}=0.348$,
the Lyapunov exponent has discontinuous partial derivatives 
$\frac{\partial\lambda(\rho,\beta;2)}{\partial\beta}$ and
$\frac{\partial\lambda(\rho,\beta;2)}{\partial\rho}$
at a point $\beta_{\rm cr}^{(2)}(\rho)$. This corresponds to a phase transition
of the Lyapunov exponent.

The phase transition appears when the supremum over $d\in [0,2]$
of $\Lambda_2(d)$ switches between two different values of $d$.
We show in Figure~\ref{Fig:ptq2} the phase transition curve 
$\beta_{\rm cr}^{(2)}(\rho)$. 
This curve ends at a critical point with coordinates 
$(\rho_c^{(2)},\beta_c^{(2)}) = (0.348,0.787)$.

An analytical approximation for the function $\Lambda_2(d)$
can be obtained in the $\beta \gg 1$ limit.
For $\beta d^2 \gg 1$ the integral is approximated by Lemma 28 in \cite{LD} as
\begin{equation}
\int_0^1 \frac{y^2 dy}{1+\rho^2 - e^{\beta d^2 (y^2-1)}} = 
\frac{1}{1+\rho^2} \left\{
\frac13 - \frac{1}{2\beta d^2} \log\left( \frac{\rho^2}{1+\rho^2} \right) + O(\beta^{-2})
\right\}.
\end{equation}
This gives the approximation valid in the region $\beta d^2 \gg 1$.
\begin{equation}\label{q2cubic}
\Lambda_2(d) = - \frac13 \beta d^3 + \beta d^2 + \frac12 d\log\frac{\rho^2}{1+\rho^2}
+ \log(1+\rho^2) .
\end{equation}

\begin{lemma}\label{lemma:cubic}
The cubic polynomial in $d$ in (\ref{q2cubic}) reaches its supremum 
at $d_{*1}=0$ or 
$d_{*2}=1+\sqrt{1+\frac{1}{2\beta}\log\frac{\rho^2}{1+\rho^2} }$, 
according to the following conditions
\begin{equation}
\sup_d \Lambda_2(d) = 
\begin{cases}
\Lambda_2(d_{*2}) & \mbox{ for }
\beta > - \frac23 \log \frac{\rho^2}{1+\rho^2}\, , \\
\Lambda_2(d_{*1}) & \mbox{ for }
\beta < - \frac23 \log \frac{\rho^2}{1+\rho^2}\,.
\end{cases}
\end{equation}
\end{lemma}

\begin{proof}
The derivative $\Lambda'_2(d)$ is a quadratic polynomial in $d$
\begin{equation}
\Lambda'_2(d) = - \beta d^2 + 2\beta d + \frac12 \log\frac{\rho^2}{1+\rho^2} .
\end{equation}
We are interested in the region $\rho < \rho_c^{(2)}<1$, so
we have $\Lambda'_2(0)=0$. For $\beta < - \frac{1}{2}
\log\frac{\rho^2}{1+\rho^2}$ we have in fact the stronger result that 
$\Lambda'_2(d)<0$ for all $d\in [0,2]$. This implies that
$\Lambda_2(d)$ is decreasing on $d=[0,2]$ and thus its supremum is reached
at $d_{*1}=0$.

For $\beta \geq - \frac{1}{2}
\log\frac{\rho^2}{1+\rho^2}$, the derivative $\Lambda'_2(d)$
becomes positive in a region $(d_1,d_2)$ centered on $d=1$, where
$d_{1,2}$ are the roots of quadratic equation $\Lambda'_2(d)=0$.
$\Lambda_2(d)$ has a local maximum at the largest root $d_2$.
This defines $d_{*2} = d_2$. (The smallest root $d_1$ 
corresponds to a minimum of $\Lambda_2(d)$.) 

Under what conditions is $d_{*2}$ a 
global supremum for $\Lambda_2(d)$? This condition can be written as
\begin{align}
\Lambda_2(d_{*2}) &= \frac23 \beta 
\left(1 + \frac{1}{2\beta} \log\frac{\rho^2}{1+\rho^2}\right) d_{*2} 
+ \frac16 \log\frac{\rho^2}{1+\rho^2} + \log(1+\rho^2) 
\\
&\geq 
\Lambda_2(d_{*1}) = \log(1+\rho^2)\,.
\nonumber
\end{align}
We used the fact that $d_{*2}$ is a solution of $\Lambda'_2(d)=0$ to eliminate
cubic and quadratic terms in $d_{*2}$ on the left side. 
This inequality can be further written as
\begin{equation}
2 R^3 + 3 R^2 -1 = (R+1)^2(2R-1) \geq 0
\end{equation}
where $R  = \sqrt{1+\frac{1}{2\beta}\log\frac{\rho^2}{1+\rho^2}}$. This
is satisfied for $R\geq \frac12$. Thus we conclude that $d_{*2}$
is a global supremum of $\Lambda_2(d)$
provided that $\beta \geq -\frac23 \log\frac{\rho^2}{1+\rho^2}$, and otherwise
the global supremum is realized at $d_{*1}=0$.
\end{proof}

\begin{figure}[t]
\centering
\includegraphics[width=4.0in]{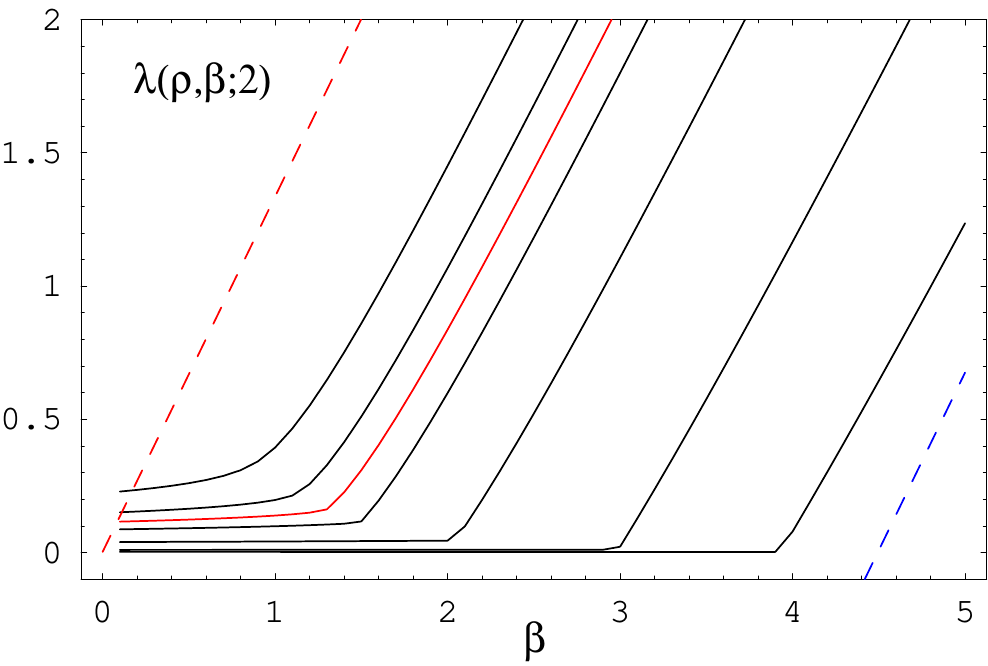}
\caption{Plots of the Lyapunov exponent $\lambda(\rho,\beta;2)$ of the $q=2$
moment of the asset price in  the Euler-Log-Euler scheme vs $\beta$ at fixed
$\rho = 0.05,0.1,0.2,0.3,0.348,0.4,0.5$ (solid curves, from bottom to top).
The solid red curve corresponds to $\rho_c^{(2)}= 0.348$.
The upper and lower bounds in (\ref{bound1}) for $\rho = 0.05$
are shown as the dashed blue and red curves, respectively.}
\label{Fig:lambda2}
\end{figure}

The $\beta \gg 1$ region corresponds to $\rho \ll 1$, which means that as 
$\rho \to 0$, the phase transition curve of the Lyapunov exponent of the $q=2$ 
moment is given approximatively by
\begin{equation}\label{q2app}
\lim_{\rho \to 0}\beta_{\rm cr}^{(2)}(\rho) = - \frac23 \log \frac{\rho^2}{1+\rho^2}\,.
\end{equation}
Along the phase transition curve, the optimizer takes the values
$d_{*1}=0, d_{*2}=\frac32$.

\subsubsection{The third moment ($q=3$)}

Proceeding in analogy with the $q=2$ moment, we take as independent
variable $d = \int_0^1 dx f(x) \in (0,2)$ instead of $h(1)$.
The Lyapunov exponent of the $q=3$ moment is given by the supremum over $d$
of the functional
\begin{align}\label{Lambda3}
\Lambda_3(d) &= \beta d^2 + \log(1+3\rho^2)\\
&\qquad\qquad  -  \beta d^3 (1+ 3\rho^2)
\int_0^1 \frac{y^2 dy}{1+3\rho^2 - e^{\beta d^2 (y^2-1)}}\,. \nonumber
\end{align}

We present in Fig.~\ref{Fig:lambda3} plots of the Lyapunov exponent of the 
third moment
$\lambda(\rho,\beta;3)$ vs $\beta$ for several values of $\rho$. The phase
transition curve $\beta_{\rm cr}^{(3)}(\rho)$ giving the points of 
discontinuity of the partial derivatives of $\lambda(\rho,\beta;3)$ 
is shown in Figure~\ref{Fig:ptq2}.

The asymptotic expression of the phase transition curve for $\rho \to 0$ can
be found in closed form, using a similar approach as for the $q=2$ case.
In the $\beta \gg 1$ limit the integral in (\ref{Lambda3}) is approximated 
as \cite{LD}
\begin{equation}
\int_0^1 \frac{y^2 dy}{1+3\rho^2 - e^{\beta d^2 (y^2-1)}} 
= \frac{1}{1+3\rho^2} \left\{
\frac13 - \frac{1}{2\beta d^2} \log\Big( \frac{3\rho^2}{1+3\rho^2} \Big) + o(\beta^{-2})
\right\}.
\end{equation}
This gives the following approximation for the functional $\Lambda_3(d)$, 
which is valid in the region $\beta d^2 \gg 1$.
\begin{equation}
\Lambda_3(d) = - \frac13 \beta d^3 + \beta d^2 + 
\frac12 d\log\frac{3\rho^2}{1+3\rho^2}
+ \log(1+3\rho^2) \,. 
\end{equation}

\begin{figure}[t]
\centering
\includegraphics[width=4.0in]{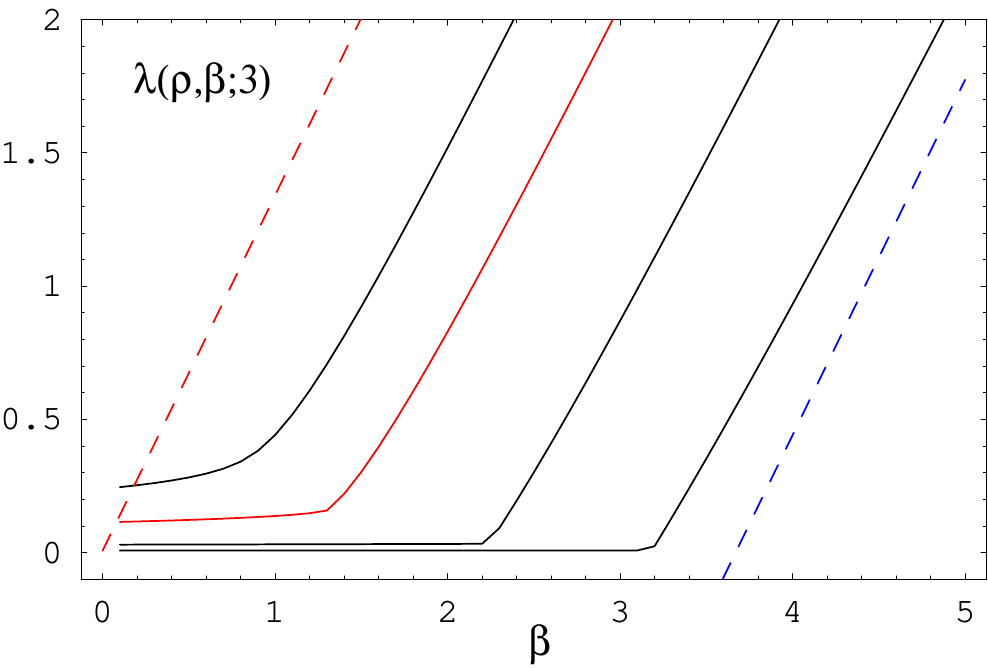}
\caption{Plots of the Lyapunov exponent $\lambda(\rho,\beta;3)$ of the $q=3$
moment of the asset price $S_n$ in the Euler-Log-Euler scheme vs $\beta$ at fixed
$\rho = 0.05,0.1,0.201,0.3$ (solid curves, from bottom to top).
The red curve corresponds to $\rho_c^{(3)}=0.201$.
The upper and lower bounds in (\ref{bound1}) for $\rho = 0.05$
are shown as the dashed red and blue curves, respectively.}
\label{Fig:lambda3}
\end{figure}

A study of this cubic polynomial in $d$ similar to that presented above
for $q=2$ in Lemma~\ref{lemma:cubic} shows that it reaches its supremum at 
$d_{*1}=0$ or 
$d_{*2}=1+\sqrt{1+\frac{1}{2\beta}\log\left(\frac{3\rho^2}{1+3\rho^2}\right)}$ 
according to the following conditions
\begin{equation}
\sup_d \Lambda_3(d) = 
\begin{cases}
\Lambda_3(d_{*2}) & \text{ for } 
\beta > - \frac23 \log \left(\frac{3\rho^2}{1+3\rho^2}\right) \, ,
\\
\Lambda_3(d_{*1}) & \text{ for } 
\beta < - \frac23 \log \left(\frac{3\rho^2}{1+3\rho^2}\right) \,.
\end{cases}
\end{equation}
This means that in the $\beta \gg 1 $ limit, or equivalently $\rho \to 0$, 
the phase transition curve of the Lyapunov exponent of the $q=3$ moment is 
given approximatively by
\begin{equation}\label{q3app}
\lim_{\rho \to 0}\beta_{\rm cr}^{(3)}(\rho) = 
- \frac23 \log \frac{3\rho^2}{1+3\rho^2}\,.
\end{equation}

\begin{table}[ht]
\caption{Critical parameters for the first few moments 
$(\rho_c^{(q)},\beta_c^{(q)})$. }
\centering 
\begin{tabular}{c |c  c} 
\hline\hline 
$q$ & $\rho_c^{(q)}$ & $\beta_c^{(q)}$ \\
\hline 
2 & 0.348 & 0.787 \\ 
3 & 0.201 & 0.787 \\
4 & 0.187 & 1.163  \\
5 & 0.154 & 1.368  \\
6 & 0.140 & 1.656  \\
7 & 0.127 & 1.913   \\
\hline 
\end{tabular}
\label{table:cr} 
\end{table}

\subsubsection{Higher Moments ($q \geq 4$)}

The phase transition of the Lyapunov exponents of the higher positive
integer moments $q \geq 4$ can be studied numerically using the solution 
of the variational problem in Proposition~\ref{prop:LambdaSimple}. 
We show in Figure~\ref{Fig:ptq2}
the phase transition curves for the moments $q=4,5,6,7$, and the corresponding
critical parameters $(\rho_c^{(q)}, \beta_c^{(q)})$ are listed in Table~\ref{table:cr}.

\subsection{Mean-field Approximation}

A simple lower bound for the Lyapunov exponent is obtained by using a constant
Ansatz for the function $f(x)=a$.  This gives a lower bound for the supremum of 
the functional $\Lambda[f]$. This assumption is equivalent to a mean-field
approximation, and it has also a phase transition. We discuss here the properties
of the phase transition in the mean-field approximation.

The functional $\Lambda_q[a]$ is given in this approximation by
\begin{equation}
\Lambda_q^{\rm MF}[a] = \log\rho a + \frac13\beta a^2 - I_q(a)\,.
\end{equation}

For $q=2$ and $q=3$ the variational problem for the mean-field approximation
leads to the Curie-Weiss theory \cite{Ellis}. 
This follows from the observation that we have
$I_2(a) = I_0(a/2)$ and $I_3(a) = I_0(a/2) - \frac12 a \log 3$, where 
$I_0(x)  = x\log x + (1-x)\log(1-x)$ is the rate function for a Bernoulli 
random variable $X$ taking values $X:(0,1)$ with probabilities $(\frac12,\frac12)$.

Upon redefinition $a/2 \to a$ and appropriate redefinition of $\rho$, the
solution of the variational problem reduces to the well-known Curie-Weiss mean-field
theory, which leads to the van der Waals equation of state for a lattice gas
with uniform interaction energy, and describes also the exact solution of an 
edge and $2-$star model in the exponential random graph models \cite{AZ,CD,RY}. 

A brief summary of the main predictions of the Curie-Weiss model can be found
for example in \cite{LD}, Proposition 9. In particular, 
the phase transition curve can be found in closed form. For 
$q=2$ this is given by
\begin{equation}\label{vdWq2}
\beta_{\rm cr}^{\rm (2) MF}(\rho) = -\frac32\log\rho
\end{equation}
and the critical point has coordinates
\begin{equation}
\rho_c^{\rm (2) MF} = e^{-1} = 0.368\,,\qquad \beta_c^{\rm (2) MF} = \frac32\,.
\end{equation}

A similar analysis for $q=3$  gives the phase transition curve
\begin{equation}\label{vdWq3}
\beta_{\rm cr}^{\rm (3) MF}(\rho) = -\frac34\log(3\rho^2)
\end{equation}
and the critical point 
\begin{equation}
\rho_c^{\rm (3) MF} = \frac{1}{\sqrt3 e} = 0.212\,,\qquad 
\beta_c^{\rm (3) MF} = \frac32\,.
\end{equation}

The numerical values of these parameters are very close to the exact 
critical parameters $(\rho_c^{(q)}, \beta_c^{(q)})$ in Table~\ref{table:cr}.
We also show the phase transition curves $\beta_{\rm cr}^{(q)}$ with $q=2,3$
in the mean-field approximation in Figure~\ref{Fig:ptq2} as the dashed blue 
curves. We note that they are close to the exact phase transition curve. 
We conclude that the mean-field approximation gives a simple qualitative 
description and represents a sufficiently accurate approximation for the
phase transition curve for the $q=2,3$ moments.

\begin{figure}[t]
\centering
\includegraphics[width=4.0in]{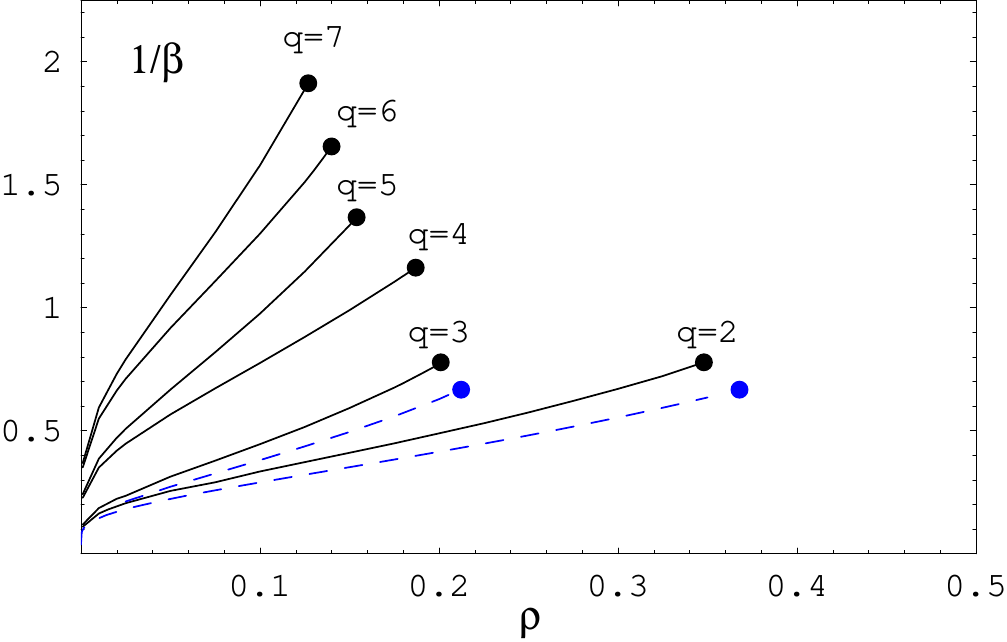}
\caption{The phase transition curves for the Lyapunov exponents of the 
$q=2,3,4,5,6,7$ moments $\mathbb{E}[(S_n)^q)]$ in the Euler-Log-Euler
discretization (black curves) in coordinates $(\rho,1/\beta)$. 
The dashed blue curves show the 
phase transition curves for $q=2,3$ in the mean-field approximation.}
\label{Fig:ptq2}
\end{figure}

\subsection{Almost Sure Limit and Fluctuations}

We have already studied the Lyapunov exponent 
$\lim_{n\rightarrow\infty}\frac{1}{n}\log\mathbb{E}[(S_{n})^{q}]$
for any positive integer $q$. In this section, we investigate the almost sure 
asymptotic behavior of $S_{n}$. Recall that
\begin{equation}
S_{n}=S_{0}\prod_{i=0}^{n-1}(1+\sigma_{0}\sqrt{\tau}e^{\omega Z_{i}-\frac{1}{2}\omega^{2}t_{i}}\varepsilon_{i}).
\end{equation}
Thus, $S_{n}$ may become negative. Hence we study the asymptotic behavior of $|S_{n}|$ instead. We have the following result.

\begin{proposition}\label{ASIThm}
The asset price $S_n$ in the Euler-Log-Euler time discretization has the
limit
\begin{equation}
\lim_{n\rightarrow\infty}\frac{1}{n}\log|S_{n}| =
\mathbb{E}\left[\log\left|1+\rho\varepsilon_{1}\right|\right]
\qquad\text{a.s.}
\end{equation}
\end{proposition}

\begin{proof}
The proof will be given in Section \ref{AppendixSection}.
\end{proof}

The dependent on $\rho$ of the expectation appearing  in this result is 
shown in Figure~\ref{Fig:expint}.

The almost sure limit is governed by the typical events, i.e., law of large numbers
and the Lyapunov exponent is governed by the rare events, i.e., large deviations. 
Since the law of large numbers for $\frac{1}{n}\log|S_{n}|$ has been established in Proposition \ref{ASIThm},
it is reasonable to study the fluctuations around the almost sure limit, i.e., central limit theorem. 
We have the following result.

\begin{proposition}\label{CLTII}
For the Euler-Log-Euler time discretization we have the following limit
\begin{align}
&\frac{\log|S_{n}|-\mathbb{E}[\log|1+\rho\varepsilon_{1}|]n}{\sqrt{n}}
\\
&\qquad\qquad\qquad\qquad\rightarrow 
N\left(0,\frac{2}{3}(H'(0))^{2}\rho^{2}\beta
+\text{Var}[\log|1+\rho\varepsilon_{1}|]\right),
\nonumber
\end{align}
in distribution as $n\rightarrow\infty$, where $H(x):=\mathbb{E}[\log|1+\frac{\varepsilon_{1}x}{1+\rho\varepsilon_{1}}|]$.
\end{proposition}

\begin{proof}
The proof will be given in Section \ref{AppendixSection}.
\end{proof}

\section{Log-Euler-Log-Euler Scheme}\label{LogEulerLogEulerSection}

\subsection{Lyapunov Exponents of the Moments}

We would like to compute the moments of the asset price in the 
Log-Euler Log-Euler time discretization.
Let us recall that this scheme is defined by the stochastic recursion
\begin{align}
&S_{i+1}=S_{i} 
\exp\left(\sigma_i (W_{i+1} - W_i) - \frac12 \sigma_i^2 \tau \right)
\\
&\sigma_{i+1}=\sigma_{i}\exp\left(\omega(Z_{i+1}-Z_{i})-\frac{1}{2}\omega^{2}\tau\right).
\nonumber
\end{align}
Therefore, 
\begin{equation}
\mathbb{E}[(S_{n})^{q}]=S_{0}^{q}\mathbb{E}\left[e^{\frac{1}{2}q(q-1)\sum_{k=0}^{n-1}\sigma_{k}^{2}\tau}\right],
\end{equation}
where $\sigma_{k}=\sigma_{0}e^{\omega Z_{k}-\frac{1}{2}\omega^{2}t_{k}}$ and
$\rho=\sigma_{0}\sqrt{\tau}$ and $\beta=\frac{1}{2}\omega^{2}n^{2}\tau$ are positive constants. 

Unlike the Euler-log-Euler scheme where we studied non-negative integer 
moments, here we can study any $q$th moments, where $q\in\mathbb{R}$.

\begin{theorem}\label{LogEulerLogEulerThm}
For any $q<0$ or $q>1$, $\mathbb{E}[(S_{n})^{q}]=\infty$.

i) For any $0\leq q\leq 1$,
\begin{align}
\lambda(\rho,\beta;q)
&=\lim_{n\rightarrow\infty}\frac{1}{n}\log\mathbb{E}[(S_{n})^{q}]
\\
&=\sup_{g\in\mathcal{AC}_{0}[0,1]}\left\{\frac{\rho^{2}q(q-1)}{2}\int_{0}^{1}e^{2\sqrt{2\beta}g(x)}dx
-\frac{1}{2}\int_{0}^{1}(g'(x))^{2}dx\right\},
\nonumber
\end{align}
where $\mathcal{AC}_{0}[0,1]$ is the space of absolutely continuous functions $f:[0,1]\rightarrow\mathbb{R}$
such that $f(0)=0$.

ii) The solution of the variational problem in the previous point is given explicitly by
\begin{equation}
\lambda(\rho,\beta;q)=\lambda\left(\frac{1}{2}\rho^{2}q(1-q),2\sqrt{2\beta}\right),
\end{equation}
where the function $\lambda(a,b)$ is given in explicit form in 
Section \ref{VarSection}, see Proposition \ref{prop:lamab}.
\end{theorem}

\begin{proof}
The proof will be given in Section \ref{AppendixSection}.
\end{proof}

\begin{figure}[t]
\centering
\includegraphics[width=4.0in]{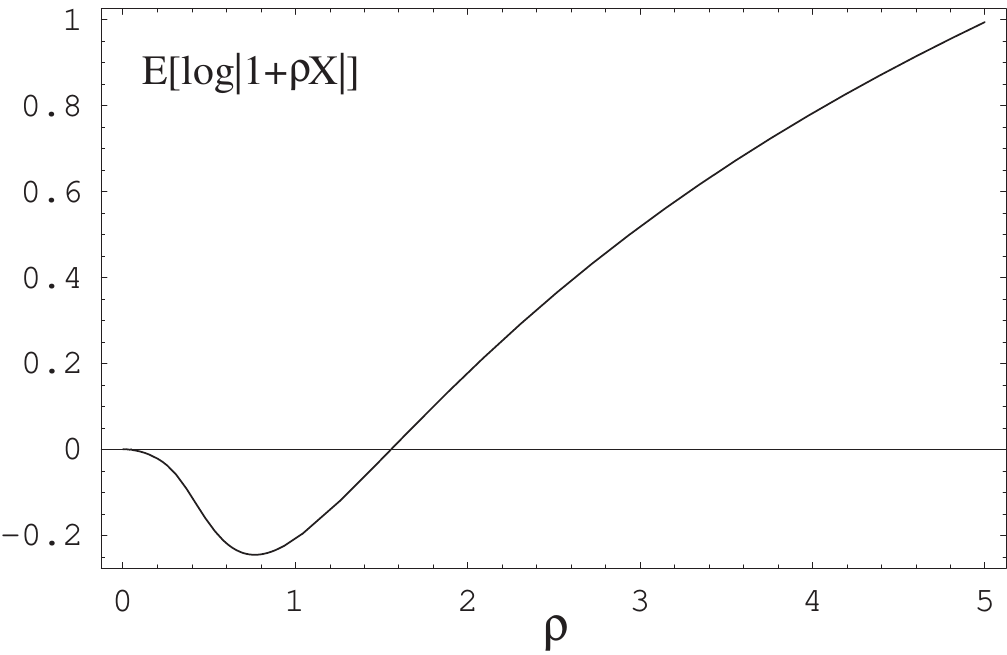}
\caption{Plot of the expectation $\mathbb{E}[\log|1+\rho \varepsilon|]$ vs. 
$\rho$ appearing in the Proposition \ref{ASIThm}, giving the law of large numbers 
for the Euler-Log-Euler discretization. This expectation is negative for
$\rho < 1.556$ and positive for $\rho > 1.556$.}
\label{Fig:expint}
\end{figure}

\subsection{Variational Problem}\label{VarSection}

We would like to solve the variational problem
\begin{equation}\label{lamdef}
\lambda(a,b) \equiv \mbox{sup}_g \left(
-a \int_0^1 e^{b g(x)} dx - \frac12 \int_0^1 (g'(x))^2 dx\right)
\end{equation}
with $a, b>0$ positive real constants. The function $g(x)$ satisfies the
boundary condition $g(0)=0$. 

\begin{proposition}\label{prop:lamab}
The solution of the variational problem (\ref{lamdef}) exists for any positive
real numbers $a,b$, and is given by 
\begin{equation}\label{lamabsol}
\lambda(a,b) = a \left( \cos^2 \xi - \frac{1}{\xi} \sin (2\xi) \right),
\end{equation}
where $\xi\in (0,\frac{\pi}{2})$ is the solution of the equation
\begin{equation}
2\xi^2 = ab^2 \cos^2 \xi\,.
\end{equation}
\end{proposition}

\begin{proof}
The proof will be given in Section \ref{AppendixSection}.
\end{proof}

The solution (\ref{lamabsol}) has a simple behavior in the limits of very 
small and very large values of $b$:
\begin{enumerate}
\item 
$\lim_{b\to 0^{+}}\lambda(a,b)=-a$;
\item
$\lambda(a,b)=-\frac{2\sqrt{2a}}{b}+\frac{\pi^{2}}{2b^{2}}+o(b^{-3})$, as $b\rightarrow\infty$;
\item
$\lambda(a,b)=-\frac{2\sqrt{2a}}{b}+o(1)$, as $a\rightarrow\infty$;
\item
$\lambda(a,b)=-a+o(1)$, as $a\rightarrow 0^{+}$.
\end{enumerate}

In our context, $a=\frac{1}{2}\rho^{2}q(1-q)$ and $b=2\sqrt{2\beta}$. 
Therefore, we have the following asymptotic results for the Lyapunov exponent 
$\lambda(\rho,\beta;q)$ in Theorem \ref{LogEulerLogEulerThm}
for the large $\rho$ and large $\beta$ limits and small $\rho$ and small 
$\beta$ limits.

\begin{proposition}
The asymptotics for the Lyapunov exponent $\lambda(\rho,\beta;q)$ in the 
Log-Euler-Log-Euler discretization for large and small $\rho$, $\beta$ are 
summarized as follows.
\begin{enumerate}
\item 
$\lim_{\beta\to 0^{+}}\lambda(\rho,\beta;q)=-\frac{1}{2}\rho^{2}q(1-q)$;
\item
$\lambda(\rho,\beta;q)=-\frac{\rho\sqrt{q(1-q)}}{\sqrt{2\beta}}+\frac{\pi^{2}}{16\beta}+o(\beta^{-3/2})$, 
as $\beta\rightarrow\infty$;
\item
$\lambda(\rho,\beta;q)=-\frac{\rho\sqrt{q(1-q)}}{\sqrt{2\beta}}+o(1)$, as $\rho\rightarrow\infty$;
\item
$\lambda(\rho,\beta;q)=-\frac{1}{2}\rho^{2}q(1-q)+o(1)$, as $\rho\rightarrow 0^{+}$.
\end{enumerate}
\end{proposition}

Plots of $\lambda(a,b)$ vs $a$ and $b$ are shown in Figure~\ref{Fig:lamexp}. 
The qualitative features of the plots agree with the properties of the 
solution discussed above.

\begin{figure}[t]
\centering
\includegraphics[width=4.0in]{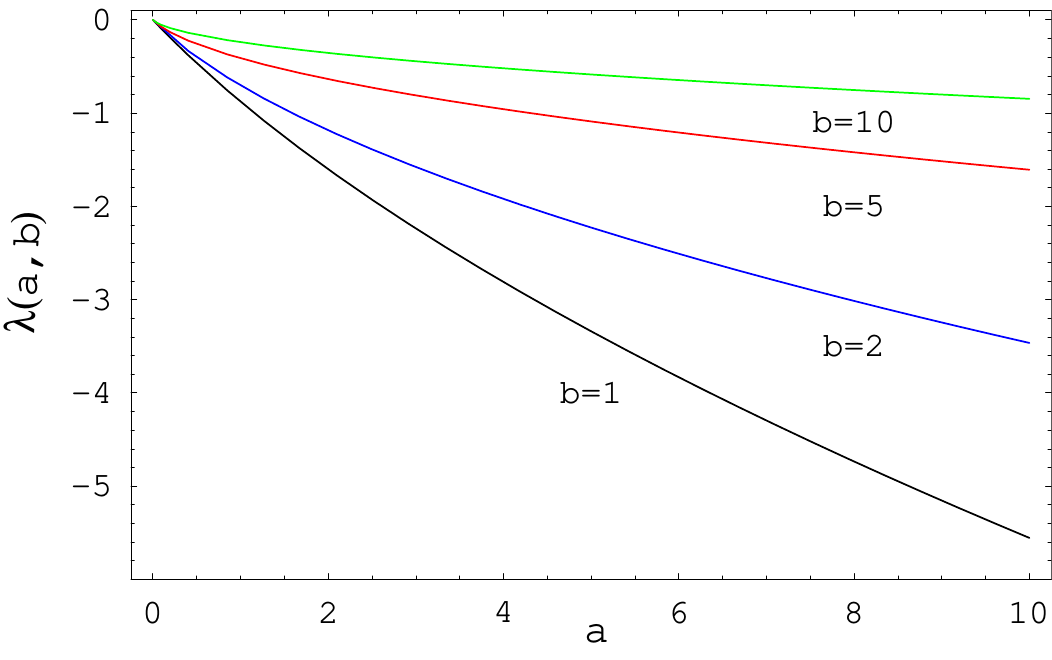}
\includegraphics[width=4.0in]{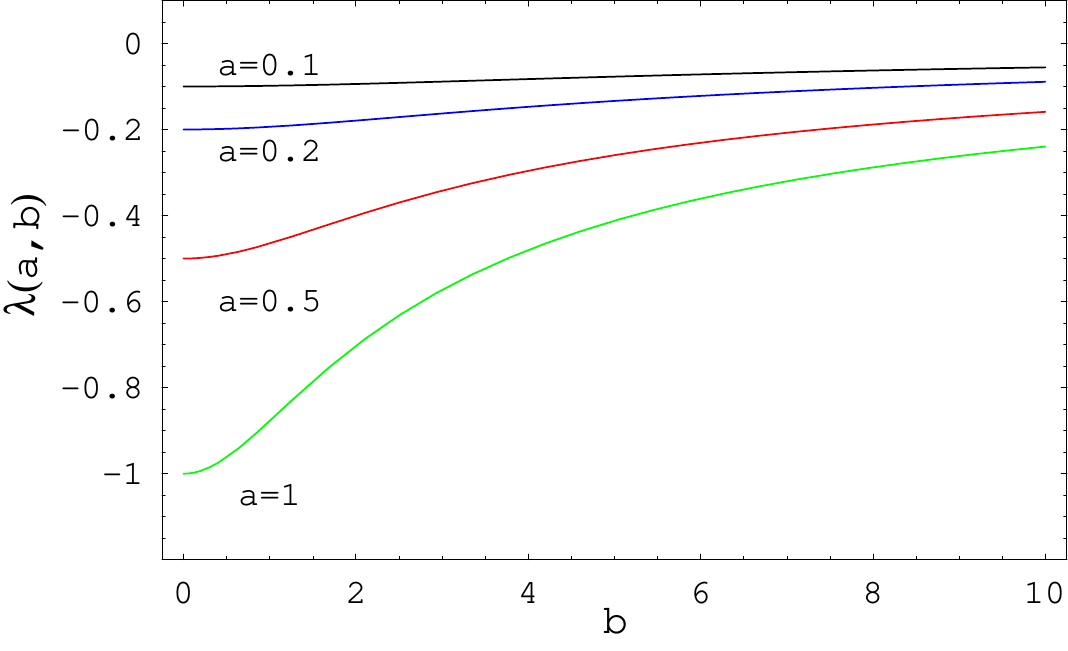}
\caption{Plots of the function $\lambda(a,b)$ given in equation (\ref{lamabsol})
vs $a$ (above) and vs $b$ (below), for several values of $(a,b)$.}
\label{Fig:lamexp}
\end{figure}


\subsection{Almost Sure Limit and Fluctuations}

\begin{proposition}\label{ASIIProp}
In the Log-Euler-Log-Euler discretization, we have the limit
\begin{equation}
\lim_{n\rightarrow\infty}\frac{1}{n}\log S_{n}=-\frac{1}{2}\rho^{2},\qquad\text{a.s.}
\end{equation}
\end{proposition}

\begin{proof}
The proof will be given in Section \ref{AppendixSection}.
\end{proof}

Since the law of large numbers for $\frac{1}{n}\log S_{n}$ has been established in Proposition \ref{ASIIProp},
it is reasonable to study the fluctuations around the almost sure limit, i.e., central limit theorem. 
We have the following result.

\begin{proposition}\label{CLTI}
In the Log-Euler-Log-Euler discretization, we have 
\begin{equation}
\frac{\log S_{n}+\frac{1}{2}\rho^{2}n}{\sqrt{n}}\rightarrow 
N\left(0,\rho^{2}+\frac{2}{3}\rho^{4}\beta\right)\,,
\end{equation}
in distribution as $n\rightarrow\infty$.
\end{proposition}

\begin{proof}
The proof will be given in Section \ref{AppendixSection}.
\end{proof}

Note that the variance in the central limit theorem in Proposition \ref{CLTI} has two terms, $\rho^{2}$
and $\frac{2}{3}\rho^{4}\beta$. When the volatility is a constant, i.e., $\sigma_{k}=\sigma_{0}$ for any $k$,
by using the identity $\rho=\sigma_{0}\sqrt{\tau}$, the asset price is given by a standard geometric Brownian motion,
\begin{equation}
S_{n}=S_{0}e^{\rho B_{n}-\frac{1}{2}\rho^{2}n},
\end{equation}
where $B_{n}$ is a standard Brownian motion and it is clear that
\begin{equation}
\frac{\log S_{n}+\frac{1}{2}\rho^{2}n}{\sqrt{n}}\rightarrow N(0,\rho^{2}).
\end{equation}
Hence, the first term $\rho^{2}$ in the variance in Proposition \ref{CLTI} can 
be interpreted as the fluctuations from the asset price and the second term 
$\frac{2}{3}\rho^{4}\beta$ explains the fluctuations of the stochastic volatility.

\section{Log-Euler-Euler Scheme}\label{LogEulerEulerSection}

\subsection{Lyapunov Exponents of the Moments}

We would like to compute the moments of the asset price in the 
Log-Euler-Euler time discretization defined by (\ref{LogEulerEulerS}).
They are given by
\begin{equation*}
\mathbb{E}[(S_{n})^{q}]
=S_{0}^{q}\mathbb{E}\left[e^{\frac{1}{2}q(q-1)\sum_{k=0}^{n-1}\sigma_{k}^{2}\tau}\right],
\end{equation*}
where $\rho=\sigma_{0}\sqrt{\tau}$, $\beta=\frac{1}{2}\omega^{2}n^{2}\tau$ and
\begin{equation}
\sigma_{k}=\sigma_{0}\prod_{j=1}^{k}(1+\omega\sqrt{\tau}V_{j}),
\end{equation}
where $V_{j}$ are i.i.d. $N(0,1)$ random variables.

\begin{theorem}\label{LogEulerEulerThm}
For any $q<0$ or $q>1$, $\lim_{n\rightarrow\infty}\frac{1}{n}\log\mathbb{E}[(S_{n})^{q}]=\infty$.

For any $0\leq q\leq 1$,
\begin{align}
&\lim_{n\rightarrow\infty}\frac{1}{n}\log\mathbb{E}[(S_{n})^{q}]
\\
&=\sup_{g\in\mathcal{AC}_{0}[0,1]}\left\{\frac{\rho^{2}q(q-1)}{2}\int_{0}^{1}e^{2\sqrt{2\beta}g(x)}dx
-\frac{1}{2}\int_{0}^{1}(g'(x))^{2}dx\right\}
\nonumber
\\
&=\lambda\left(\frac{1}{2}\rho^{2}q(1-q),2\sqrt{2\beta}\right),
\nonumber
\end{align}
where the function $\lambda(a,b)$ is given explicitly in Section \ref{VarSection}, see Proposition \ref{prop:lamab}.
\end{theorem}

\begin{proof}
The proof will be given in Section \ref{AppendixSection}.
\end{proof}

\begin{lemma}\label{EquivLemma}
For any $\epsilon>0$,
\begin{equation}
\limsup_{n\rightarrow\infty}\frac{1}{n}
\log\mathbb{P}\left(\sup_{0\leq x\leq 1}\left|\frac{1}{n}\sum_{j=1}^{\lfloor xn\rfloor}V_{j}^{(n)}-\frac{1}{n}\sum_{j=1}^{\lfloor xn\rfloor}V_{j}\right|
\geq\epsilon\right)=-\infty.
\end{equation}
\end{lemma}

\begin{proof}
The proof will be given in Section \ref{AppendixSection}.
\end{proof}


\subsection{Almost Sure Limit and Fluctuations}

In the Log-Euler-Euler scheme we have
\begin{equation}
S_{n}=S_{0}e^{\sum_{i=0}^{n-1}\sigma_{i}\Delta W_{i}-\frac{1}{2}\sigma_{i}^{2}\tau},
\end{equation}
where $\sigma_{i}=\sigma_{0}\prod_{j=1}^{i}(1+\omega\sqrt{\tau}V_{j})$.
One can show that $\sigma_{n}\rightarrow\sigma_{0}$ and thus we have 
the following almost sure limit.

\begin{proposition}\label{ASIIIProp}
\begin{equation}
\lim_{n\rightarrow\infty}\frac{1}{n}\log S_{n}=-\frac{1}{2}\rho^{2},\qquad\text{a.s.}
\end{equation}
\end{proposition}

\begin{proof}
We omit the details of the proof.
\end{proof}

The corresponding result for fluctuations is given by the following proposition.

\begin{proposition}
\begin{equation}
\frac{\log S_{n}+\frac{1}{2}\rho^{2}n}{\sqrt{n}}\rightarrow 
N\left(0,\rho^{2}+\frac{2}{3}\rho^{4}\beta\right),
\end{equation}
in distribution as $n\rightarrow\infty$.
\end{proposition}

\begin{proof}
We omit the details of the proof.
\end{proof}

\section{Euler-Euler Scheme}\label{EulerEulerSection}

\subsection{Lyapunov Exponents of the Moments}

We would like to compute the moments of the asset price in the 
Euler-Euler discretization defined by (\ref{EulerEulerS}).
They are given by
\begin{align}
\mathbb{E}[(S_{n})^{q}]
&=S_{0}^{q}\mathbb{E}\left[\prod_{k=0}^{n-1}\left(1+\sigma_{0}\sqrt{\tau}\prod_{j=1}^{k}(1+\omega\sqrt{\tau}V_{j})\varepsilon_{k}\right)^{q}\right]
\\
&=S_{0}^{q}m(q)^{n}\mathbb{E}
\left[\prod_{k=0}^{n-1}e^{\log\rho Y_{k}+(\frac{\sqrt{2\beta}}{n}\sum_{j=1}^{k}V_{j}^{(n)})Y_{k}}\right],
\nonumber
\end{align}
where $\rho=\sigma_{0}\sqrt{\tau}$, $\beta=\frac{1}{2}\omega^{2}n^{2}\tau$, and
\begin{equation}
V_{j}^{(n)}:=n\frac{1}{\sqrt{2\beta}}\log\left|1+\frac{\sqrt{2\beta}}{n}V_{j}\right|.
\end{equation}
$(Y_{k})_{k=0}^{n-1}$ are i.i.d. random variables such that
\begin{equation}
\mathbb{P}(Y_{k}=2j)=\frac{1}{m(q)}\frac{q!}{(2j)!(q-2j)!}(2j-1)!!,
\qquad
j=0,1,2,\ldots,[q/2].
\end{equation}
$m(q)$ is the normalizer defined as
\begin{equation}
m(q)=\sum_{j=0}^{[q/2]}\frac{q!}{(2j)!(q-2j)!}(2j-1)!!.
\end{equation}

\begin{theorem}\label{EulerEulerThm}
\begin{align}
&\lim_{n\rightarrow\infty}\frac{1}{n}\log\mathbb{E}[(S_{n})^{q}]
\\
&=\sup_{g\in\mathcal{G}_{q}}
\left\{g(1)\log\rho+\beta\int_{0}^{1}(g(1)-g(x))^{2}dx-\int_{0}^{1}I_{q}(g'(x))dx\right\}.
\nonumber
\end{align}
\end{theorem}

\begin{proof}
The proof will be given in Section \ref{AppendixSection}.
\end{proof}

\subsection{Almost Sure Limit and Fluctuations}

In the Euler-Euler scheme we have
\begin{equation}
S_{n}=S_{0}\prod_{k=0}^{n-1}\left(1+\sigma_{0}\sqrt{\tau}\prod_{j=1}^{k}(1+\omega\sqrt{\tau}V_{j})\varepsilon_{k}\right).
\end{equation}
One can show that
\begin{equation}
\sigma_{n}=\sigma_{0}\prod_{j=1}^{n}(1+\omega\sqrt{\tau}V_{j})\rightarrow\sigma_{0}
\end{equation}
a.s. as $n\rightarrow\infty$ and then follow the similar arguments as before, one can show that


\begin{proposition}\label{ASIIIIProp}
\begin{equation}
\lim_{n\rightarrow\infty}\frac{1}{n}\log|S_{n}|=
\mathbb{E}\left[\log\left|1+\rho\varepsilon_{1}\right|\right]\qquad\text{a.s.}
\end{equation}
\end{proposition}

\begin{proof}
We omit the details of the proof.
\end{proof}

\begin{proposition}
\begin{align}
&\frac{\log|S_{n}|-\mathbb{E}[\log|1+\rho\varepsilon_{1}|]n}{\sqrt{n}}
\\
&\qquad\qquad\qquad\qquad\rightarrow 
N\left(0,\frac{2}{3}(H'(0))^{2}\rho^{2}\beta
+\text{Var}[|1+\rho\varepsilon_{1}|]\right),
\nonumber
\end{align}
in distribution as $n\rightarrow\infty$, where $H(x):=\mathbb{E}[\log|1+\frac{\varepsilon_{1}x}{1+\rho\varepsilon_{1}}|]$.
\end{proposition}

\begin{proof}
We omit the details of the proof.
\end{proof}


\section{Comparison with the Continuous-time Model and Applications}
\label{sec:7} 

In this section we compare the predictions of the four time discretizations
discussed above against the known properties of the continuous-time
stochastic volatility model (\ref{SDES}). As mentioned in the Introduction,
the discrete time model is not expected to reproduce the properties of the
continuous time model.
We also discuss the implications of 
our results for the numerical simulation of the model in a discrete time setting. 

\subsection{Martingale property}

The time discretized model in the Euler-Maruyama discretization has the following
property, which is in marked contrast with the corresponding property of the
continuous time model.

\begin{proposition}\label{prop:martingale}
Under all four applications of the Euler-Maruyama discretization, the 
asset price is a martingale, for any value of the correlation 
$\varrho \in [-1,1]$
\begin{eqnarray}\label{MartingaleDiscrete}
\mathbb{E}[S_n] = S_0 \,.
\end{eqnarray}
\end{proposition}

\begin{proof}

i) Euler-Log Euler discretization. We have
\begin{eqnarray}
S_{i+1} = S_i ( 1 + \varrho \sigma_i  \Delta Z_i + \sqrt{1-\varrho^2}
\sigma_i \Delta W_i^\perp )
\end{eqnarray}
with $\sigma_i = \sigma_0 e^{\omega Z_i - \frac12\omega^2 t_i}$, 
$\Delta Z_i = Z_{i+1} - Z_i$ and $W_i^\perp$ is a standard Brownian
motion uncorrelated with $Z_i$. It is easy to see that one has
$\mathbb{E}[S_{i+1}| \mathcal{F}_i] = S_i$.
Repeated application of the tower relation for the conditional expectations
gives (\ref{MartingaleDiscrete}).

ii) Log  Euler-Log Euler. We have
\begin{eqnarray}
S_{i+1} = S_i e^{\varrho\sigma_i \Delta Z_i + \sqrt{1-\varrho^2} \sigma_i \Delta W_i^\perp
- \frac12 \sigma_i^2 \tau} \,.
\end{eqnarray}
We have again the relation $\mathbb{E}[S_{i+1}| \mathcal{F}_i] = S_i$
which leads to (\ref{MartingaleDiscrete}) by repeated application of the 
tower relation for conditional expectations.

The same result is obtained also for the Log Euler-Euler, and Euler-Euler
schemes, and we omit the proofs. 
\end{proof}

\subsection{Large-Maturity Limit}

The properties of the continuous-time model (\ref{SDES}) in the large-maturity
limit $t\to \infty$ have been studied in \cite{Forde1}. The asset price $S_t$
converges in distribution as $t\to \infty$ \cite{Forde1}
\begin{equation}\label{Sinfty}
\lim_{t\to \infty}  S_t = 
\begin{cases}
S_\infty  &  \omega > 0 \\
0         &  \omega=0 
\end{cases}
\,.
\end{equation}
$S_\infty$ is a random variable which has
$\mathbb{P}(S_\infty>0)>0$ and is known in closed form. 

The corresponding limit result for the growth rate of $\log S_t$ is
\begin{equation}
\lim_{t\to \infty} \frac{1}{t} \log S_t = 
\begin{cases}
0  & \omega > 0 \\
- \frac12\sigma_0^2 &  \omega=0 
\end{cases}
\,.
\end{equation}
The discontinuous behavior for $\omega=0$ and $\omega>0$ is due to the different
limits for the volatility $\sigma_t$ in the $t\to \infty$ limit. If $\omega=0$ 
the volatility is constant $\sigma_t = \sigma_0$, while for any positive
non-zero value of the parameter $\omega>0$, the limit is 
$\lim_{t\to\infty}\sigma_t = 0$.

We would like to compare these results with the large-time limit in the
time discretization considered in this paper, keeping in mind that
the $n\to \infty$ limit considered here is taken at fixed $\beta$.
This means that $\omega$ is small and approaches zero as
\begin{equation}
\omega = \frac{\sqrt{2\beta}}{n \sqrt{\tau}} \to 0\,.
\end{equation}
This is different from the $t\to \infty$ limit in \cite{Forde1},
which is taken at fixed $\omega>0$. 

The time discretizations considered above give different
results in the $n\to \infty$ limit at fixed time step $\tau$:

i) Euler-Log-Euler (Proposition~\ref{ASIThm}): 
\begin{equation}
\lim_{n\to \infty} \frac{1}{n} \log |S_n| = \mathbb{E}[\log | 1+\rho \varepsilon_1|]
\end{equation}
This expectation is negative for $\rho < 1.556$, and positive for 
$\rho>1.556$, as seen from Fig.~\ref{Fig:expint}.
We get thus 
\begin{equation}
\lim_{n\to \infty} S_n = 
\begin{cases}
0 & \rho < 1.556 \\
+\infty &  \rho > 1.556 \\
\end{cases}
\, .
\end{equation}
For both cases the result is different from the long maturity asymptotics
in (\ref{Sinfty}). We note that the behavior in the first case $\rho < 1.556$ 
is closer to that in the Black-Scholes model.
The condition $\rho = \sigma_0 \sqrt{\tau} < 1.556$ is satisfied in all
practical time discretizations. 

ii) Log-Euler-Log-Euler (Proposition~\ref{ASIIProp}):
\begin{equation}
\lim_{n\to \infty} \frac{1}{n} \log S_n = - \frac12 \sigma_0^2 \tau ,
\end{equation}
which gives
\begin{equation}
\lim_{n\to \infty} S_n = 0\,,
\end{equation}
similar to the asymptotics in the Euler-Log-Euler scheme for $\rho < 1.556$.

iii) Log-Euler-Euler (Proposition~\ref{ASIIIProp}):
\begin{equation}
\lim_{n\to \infty} \frac{1}{n} \log S_n = - \frac12 \sigma_0^2 \tau .
\end{equation}
This case is similar to (ii).

iv) Euler-Euler (Proposition~\ref{ASIIIIProp}):
\begin{equation}
\lim_{n\to \infty} \frac{1}{n} \log |S_n| = \mathbb{E}[\log | 1+\rho \varepsilon_1|] .
\end{equation}
This case is similar to (i). 


\subsection{Moment Finiteness}
In the uncorrelated stochastic volatility model (\ref{SDES}) in continuous time
all moments of order $q>1$ are infinite for any time $t > 0$
\cite{Jourdain,LionsMusiela}, see (\ref{Tstaru}). 

\begin{figure}
    \centering
   \includegraphics[width=4.0in]{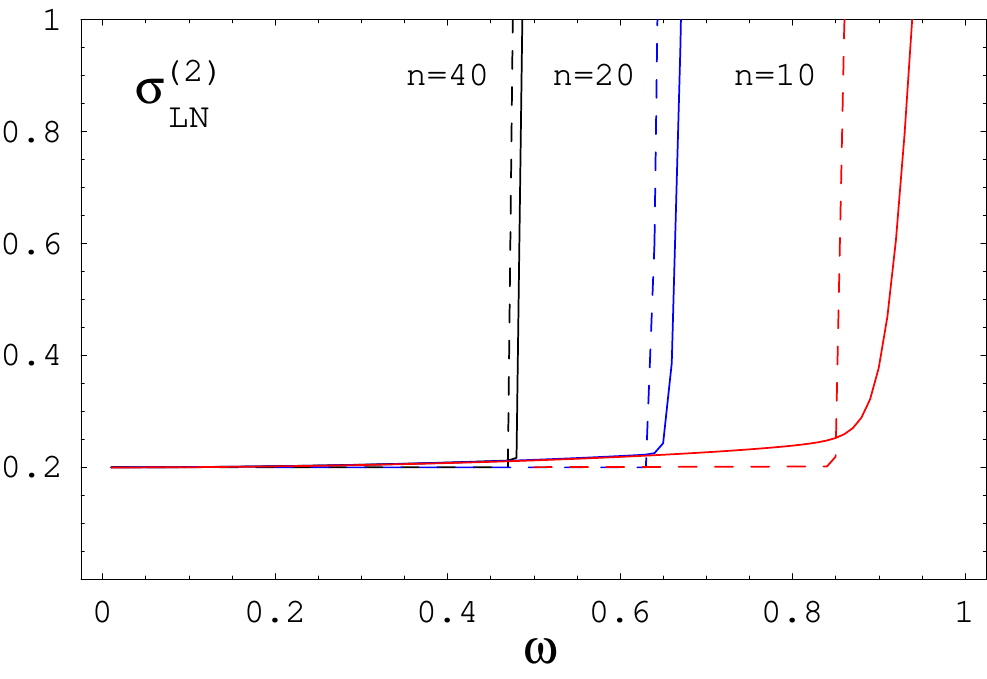}
   \includegraphics[width=4.0in]{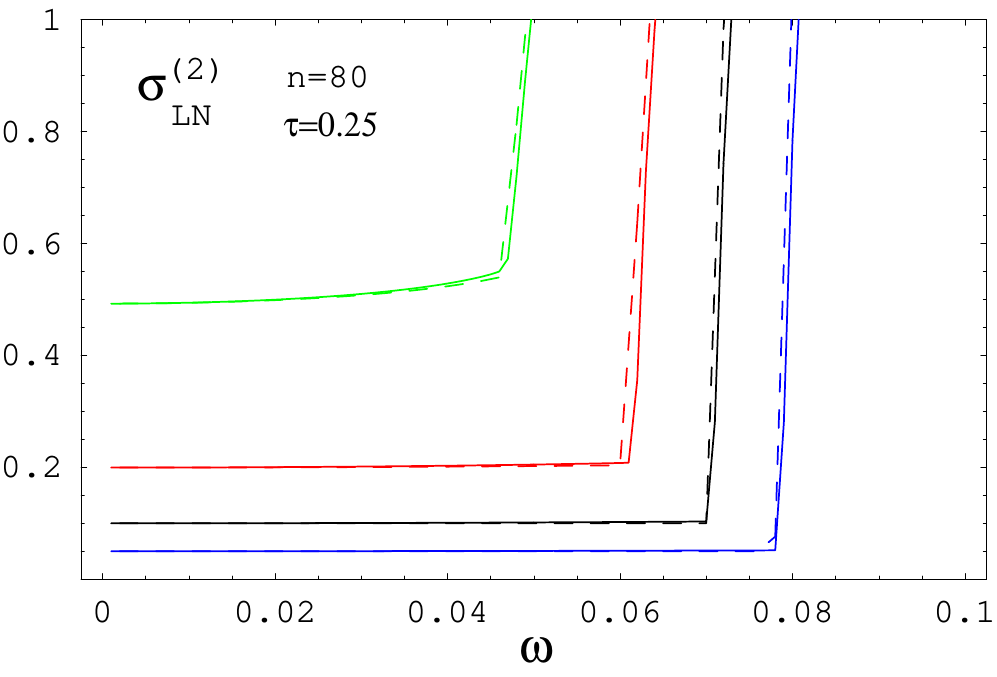}
    \caption{
Plots of the log-normal equivalent volatility of the second 
moment of the asset price in the Euler-Log-Euler scheme $\sigma_{\rm LN}^{(2)}$ 
vs $\omega$. Solid curves: finite $n$ results, dashed curves: asymptotic 
results using (\ref{sigLNasympt}) in terms of the Lyapunov exponent $\lambda(\rho,\beta;2)$.
Above: Fixed $\sigma_0 = 20\%$ and $t_n=1$ for three time discretizations with 
$n=10,20,40$ time steps (from right to left).
Below: Fixed $t_n=20$ and $\sigma_0 = 5\%, 10\%, 20\%, 50\%$
(from bottom to top). The time step is $\tau=0.25$.
}
\label{Fig:2}
 \end{figure}

In the time discretized version, this property holds only in the 
Log-Euler-Log-Euler discretization,
which is thus closest to the continuous-time model in this respect.
In particular, the variance of the asset price $S_t$ in the
Log-Euler-Log-Euler scheme is infinite for any $t>0$.
This may be inconvenient in Monte Carlo simulations of the model, where a
finite payoff variance is required for a reliable error estimate of expectation
values. If a finite variance for $S_n$ is needed, then the Euler-Log-Euler discretization may be more
appropriate, as all positive integer moments $\mathbb{E}[(S_n)^q]$
are finite in this scheme. This feature must be balanced against the potential 
inconvenience arising from the fact that the asset price $S_n$ is not positive definite.

In the Euler-Log-Euler discretization, the moments $q>1$ of the asset price 
$S_n$ are finite, but their numerical values can explode to very large values, 
exceeding double precision. 
We illustrate the explosion of the $q=2$ moment of the asset price $S_n$ in 
the Euler-Log-Euler scheme in Figure~\ref{Fig:2},
which shows plots of the equivalent log-normal volatility 
$\sigma_{\rm LN}^{(2)}(t_n)$. We define the equivalent log-normal
volatility of the $q-$th moment at maturity $t_n$ as
\begin{equation}
\sigma_{\rm LN}^{(q)2}(t_n) = \frac{2}{q(q-1)}\frac{1}{t_n} 
\log \frac{\mathbb{E}[(S_n)^q]}{S_0^q}\,.
\end{equation}
These quantities are defined such that $\sigma_{\rm LN}^{(q)}(t_n)$ is the
standard deviation of $\log S_n$ assuming that this random variable is log-normally 
distributed. This assumption is an exact result for $\omega=0$, 
which gives $\lim_{\omega \to 0} \sigma_{\rm LN}^{(q)}(t_n) = \sigma_0$ for
all $q\in \mathbb{Z}_+$. 
In the $n\to \infty$ limit, $(\sigma_{\rm LN}^{(q)}(t_n))^2$ is related to 
the Lyapunov exponent of the $q$ moment as 
\begin{equation}\label{sigLNasympt}
\lim_{n \to \infty} \sigma_{\rm LN}^{(q)2}(t_n) = \frac{2}{q(q-1)}
\frac{1}{\tau}\lambda(\rho,\beta;q)\,.
\end{equation}

\begin{figure}
    \centering
   \includegraphics[width=4.0in]{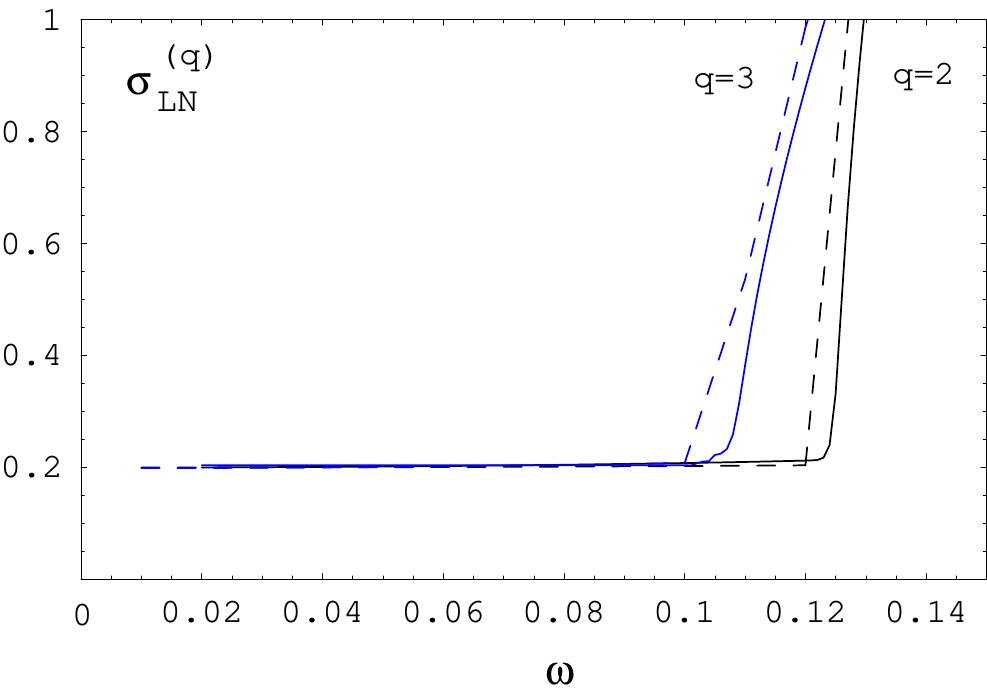}
    \caption{
Plots of the log-normal equivalent volatility of the second 
moment $\sigma_{\rm LN}^{(2)}$ (black curves) and of the third moment 
$\sigma_{\rm LN}^{(3)}$ (blue curves)
of the asset price in the Euler-Log-Euler scheme $\sigma_{\rm LN}^{(2)}$ 
vs $\omega$. Solid curves: finite $n$ results, dashed curves: asymptotic 
results using (\ref{sigLNasympt}) in terms of the Lyapunov exponent $\lambda(\rho,\beta;2)$.
Fixed $\sigma_0 = 20\%$ and $\tau=0.25, n=40$.
}
\label{Fig:3}
 \end{figure}

The positive integer moments of the asset price $S_n$ in the Euler-Log-Euler 
scheme can be computed exactly for finite $n$ using Proposition~\ref{prop:Mn}.
We used this method for the numerical
evaluation of the solid curves in Figures~\ref{Fig:2} and \ref{Fig:3}. 
From these plots one
observes that the $q=2$ moment has a rapid increase at a certain value of 
the model parameter $\omega$. This phenomenon is associated with the rapid
increase of the Lyapunov exponents $\lambda(\rho,\beta;q)$ 
when crossing the phase transition curve $\beta_{\rm cr}^{(q)}(\rho)$
from smaller to larger values of $\beta$.
For $n\gg 1$ the moments are given approximatively by
\begin{equation}
\mathbb{E}[(S_n)^q] = e^{\lambda(\rho,\beta;q) n + o(n)}\,.
\end{equation}
When crossing the phase transition curve in the direction of increasing $\beta$,
the Lyapunov exponent is continuous but it starts growing very rapidly with $\beta$.
See for example the plots in Fig.~\ref{Fig:lambda2} for $q=2$. This appears 
in numerical simulations as a numerical explosion of the respective moment. 

To illustrate this point, we show in Figure~\ref{Fig:2} (dashed curves)
also the $n\to \infty$ asymptotic
results for $\sigma_{\rm LN}^{(2)}(t_n)$ which are given by the Lyapunov exponent
$\lambda(\rho,\beta;2)$ in terms of the relation (\ref{sigLNasympt}). They are
seen to agree very well with the finite $n$ results (solid curves). The agreement
becomes better as $n$ is larger. 

The asymptotic result (\ref{sigLNasympt}) allows one to compute the explosion
value of $\omega$ in terms of the critical curve of the Lyapunov exponent 
$\lambda(\rho,\beta;2)$.
For example, the red curve in the lower plot of Figure~\ref{Fig:2} corresponds
to $\sigma_0=0.2,\tau=0.25$ and thus $\rho=\sigma_0\sqrt{\tau}  =0.1$. 
The phase transition takes place at 
$\beta_{\rm cr}^{(2)}(\rho)$ which is given by the
$q=2$ phase transition curve in Figure~\ref{Fig:ptq2}. 
As discussed in Sec.~\ref{sec:q2}, this is given to a good approximation by 
the mean-field result (\ref{q2app}) and is $\beta_{\rm cr}^{(2)}(0.1) = 3.07$. 
The value of $\omega$
at which the phase transition curve is crossed is $\omega_{\rm cr}^2  = 
\frac{2\beta_{\rm cr}^{(2)}}{t_n n} = (0.062)^2$, which agrees very well
with the explosion value of this parameter observed in Figure~\ref{Fig:2}.

The same explosive phenomenon appears for all positive integer moments in the 
Euler-Log-Euler scheme. This is illustrated in Figure~\ref{Fig:3} where we 
compare the log-normal equivalent volatilities of the $q=2$ and $q=3$ moments. 
At given $\rho$, higher moments explode at smaller values of $\beta$, as seen
from the relative position of the phase transition curves in 
Figure~\ref{Fig:ptq2}.

In conclusion, in order to avoid the numerical explosion of the 
$q$-th moment in the Euler-Log-Euler scheme, the simulation of the model 
must be restricted to the region 
of sufficiently small values of $\beta$ such that the corresponding phase 
transition curve $\beta_{\rm cr}^{(q)}(\rho)$ is not crossed from above. 
For fixed simulation time step $\tau$ this is equivalent to restricting $\omega$
sufficiently small values, below a critical value. In general this gives a 
constraint on the simulation parameters $(\tau, \sigma_0, \omega, n)$ which 
has to be satisfied for all time steps of the simulation.


\section{Appendix: Technical Proofs}\label{AppendixSection}

\subsection{Proofs of the Results in Section \ref{EulerLogEulerSection}}

\begin{proof}[Proof of Theorem \ref{EulerLogEulerThm}]
Recall that $Z_{i+1}-Z_{i}=\sqrt{\tau}V_{i}$, where $V_{i}$ are i.i.d. $N(0,1)$ random variables.
Therefore,
\begin{align}
\mathbb{E}\left[\prod_{k=0}^{n-1}e^{\log\rho Y_{k}+\omega Z_{k}Y_{k}}\right]
&=\mathbb{E}\left[e^{\log\rho\sum_{k=0}^{n-1}Y_{k}+\sum_{k=0}^{n-1}(\sum_{j=0}^{k-1}\sigma\sqrt{\tau}V_{j})Y_{k}}\right]
\\
&=\mathbb{E}\left[e^{\log\rho\sum_{i=0}^{n-1}Y_{k}+\sum_{j=0}^{n-2}(\sum_{k=j+1}^{n-1}Y_{k})\omega\sqrt{\tau}V_{j}}\right]
\nonumber
\\
&=\mathbb{E}\left[e^{\log\rho\sum_{i=0}^{n-1}Y_{k}+\frac{\beta}{n^{2}}\sum_{j=0}^{n-2}(\sum_{k=j+1}^{n-1}Y_{k})^{2}}\right]
\nonumber
\end{align}
By large deviations theory in probability, the Mogulskii theorem (see e.g. \cite{Dembo}) says that
$\mathbb{P}(\frac{1}{n}\sum_{i=1}^{\lfloor n\cdot\rfloor}Y_{i}\in\cdot)$ satisfies a sample path large deviation principle
on the space $L_{\infty}[0,1]$ (i.e. the space of functions on $[0,1]$ equipped with supremum norm)
with the rate function
\begin{equation}
\int_{0}^{1}\mathcal{I}_{q}(g'(x))dx,
\end{equation}
where $g(0)=0$, $g$ is absolutely continuous, $0\leq g'\leq 1$, and the rate function is $+\infty$ otherwise
and $\mathcal{I}_{q}$ is a relative entropy function given by
\begin{align}
\mathcal{I}_{q}(x)
&=\sup_{\theta\in\mathbb{R}}\left\{\theta x-\log\mathbb{E}[e^{\theta Y_{1}}]\right\}
\\
&=\sup_{\theta\in\mathbb{R}}\left\{\theta x-\log\sum_{j=0}^{[q/2]}\frac{q!}{(2j)!(q-2j)!}(2j-1)!!e^{2j\theta}\right\}+\log m(q),
\nonumber
\end{align}
Informally speaking, it says that
\begin{equation}
\mathbb{P}\left(\frac{1}{n}\sum_{i=1}^{\lfloor nx\rfloor}Y_{i}\simeq g(x), 0\leq x\leq 1\right)
\simeq e^{-n\int_{0}^{1}\mathcal{I}_{q}(g'(x))dx+o(n)},
\end{equation}
as $n\rightarrow\infty$.

In large deviations theory, the celebrated Varadhan's lemma says that if $P_{n}$ satisfies
a large deviation principle with rate function $\mathcal{I}(x)$ on $\mathbb{X}$ and $F:\mathbb{X}\rightarrow\mathbb{R}$
is a bounded and continuous function, then \cite{VaradhanII}
\begin{equation}
\lim_{n\rightarrow\infty}\frac{1}{n}\log\int_{\mathbb{X}}e^{nF(x)}dP_{n}(x)=\sup_{x\in\mathbb{X}}\{F(x)-\mathcal{I}(x)\}.
\end{equation}
It is easy to check that for any $g\in L_{\infty}[0,1]\cap\mathcal{G}$,
\begin{equation}
g\mapsto\log\rho\cdot g(1)+\beta\int_{0}^{1}(g(1)-g(x))^{2}dx
\end{equation}
is a bounded and continuous map. Moreover,
\begin{align}
&n\left[\log\rho\left(\frac{1}{n}\sum_{i=1}^{n}Y_{i}\right)
+\beta\int_{0}^{1}\left(\frac{1}{n}\sum_{i=1}^{n}Y_{i}-\frac{1}{n}\sum_{i=1}^{\lfloor nx\rfloor}Y_{i}\right)^{2}dx\right]
\\
&=\log\rho\sum_{i=1}^{n}Y_{i}+\frac{\beta}{n}\int_{0}^{1}\left(\sum_{i=\lfloor nx\rfloor+1}^{n}Y_{i}\right)^{2}dx
\nonumber
\\
&=\log\rho\sum_{i=1}^{n}Y_{i}+\frac{\beta}{n^{2}}\sum_{j=0}^{n-1}\left(\sum_{i=j+1}^{n}Y_{i}\right)^{2},
\nonumber
\end{align}
whose difference from $\log\rho\sum_{i=0}^{n-1}Y_{i}+\frac{\beta}{n^{2}}\sum_{j=0}^{n-2}(\sum_{i=j+1}^{n-1}Y_{i})^{2}$
can be bounded by a deterministic constant.

Hence, by Varadhan's lemma, we conclude that
\begin{align}
\lambda(\rho,\beta;q)&=\lim_{n\rightarrow\infty}\frac{1}{n}\log\mathbb{E}[(S_{n})^{q}]
\\
&=\sup_{g\in\mathcal{G}_{q}}\left\{\log\rho g(1)+\beta\int_{0}^{1}(g(1)-g(x))^{2}dx-\int_{0}^{1}I_{q}(g'(x))dx\right\},
\nonumber
\end{align}
where $I_{q}(x)$ was defined in \eqref{IqFormula} and $\mathcal{G}_{q}$ was defined in \eqref{GqFormula}.
\end{proof}

\begin{proof}[Proof of Proposition \ref{AsympProp}]
(i) Since $x\to I_{q}(x)$ is a convex function, 
Jensen's inequality implies the upper bound
\begin{align}
\lambda(\rho,0;q)
&\leq\sup_{g\in\mathcal{G}_{q}}
\left\{\log\rho g(1)-I_{q}\left(\int_{0}^{1}g'(x)dx\right)\right\}
\\
&=\sup_{g\in\mathcal{G}_{q}}\left\{\log\rho g(1)-I_{q}(g(1))\right\}
\nonumber
\\
&=\sup_{0\leq x\leq 2[q/2]}\left\{\log\rho x-I_{q}(x)\right\}. \nonumber
\end{align}
On the other hand, choosing $g(x)=g(1)x$, it is clear that we have the lower bound
\begin{equation}
\lambda(\rho,0;q)\geq\sup_{0\leq x\leq 2[q/2]}\left\{x\log\rho-I_{q}(x)\right\}.
\end{equation}
Therefore,
\begin{equation}
\lambda(\rho,0;q)=\sup_{0\leq x\leq 2[q/2]}\left\{\log\rho x-I_{q}(x)\right\}
=f_{q}(\log\rho)
\end{equation}

(ii) We get a lower bound by taking $g(x)=2[q/2]x$,
\begin{align}
\lambda(\rho,\beta;q)
&\geq
2[q/2]\log\rho+\beta(2[q/2])^{2}\int_{0}^{1}(1-x)^{2}dx-I_{q}(2[q/2])
\\
&=2[q/2]\log\rho+\beta\frac{4}{3}[q/2]^{2}-I_{q}(2[q/2]),
\nonumber
\end{align}
which gives the lower bound in (\ref{bound1}).

On the other hand, by Mean Value Theorem, 
$|g(1)-g(x)|\leq 2[q/2]|1-x|$ for any $0\leq g'\leq 2[q/2]$. Also, using the Jensen
inequality for the last term as in (i), we get the upper bound
\begin{align}
&\lambda(\rho,\beta;q)
\\
&=
\sup_{g(0)=0,0\leq g'\leq 2[q/2]}
\left\{\log\rho g(1)+\beta\int_{0}^{1}(g(1)-g(x))^{2}dx
-\int_{0}^{1}I_{q}(g'(x))dx\right\}
\nonumber
\\
&\leq\sup_{g(0)=0,0\leq g'\leq 2[q/2]}
\left\{\log\rho g(1)+4\beta[q/2]^{2}\int_{0}^{1}(1-x)^{2}dx
-\int_{0}^{1}I_{q}(g'(x))dx\right\}
\nonumber
\\
&=\frac{4}{3}\beta[q/2]^{2}
+\sup_{0\leq x\leq 2[q/2]}\left\{x\log\rho-I_{q}(x)\right\}.\nonumber
\end{align}
This proves the upper bound in (\ref{bound1}).

(iii)
Dividing by $\beta$ in (\ref{bound1})
and taking the $\beta\to\infty$ limit implies the relation (\ref{LargeBeta}).
Moreover, 
\begin{align}\label{DifferenceBound}
&\sup_{0\leq x\leq 2[q/2]}\left\{x\log\rho-I_{q}(x)\right\}-\{2[q/2]\log\rho-I_{q}(2[q/2])\}
\\
&=f_{q}(\log\rho)-\{2[q/2]\log\rho-I_{q}(2[q/2])\}
\nonumber
\\
&=\log\left(\sum_{j=0}^{[q/2]}\frac{q!}{(2j)!(q-2j)!}(2j-1)!!\rho^{2j}\right)-\{2[q/2]\log\rho-I_{q}(2[q/2])\}
\nonumber
\\
&=\log\left(\frac{q!}{(2[q/2])!(q-2[q/2])!}(2[q/2]-1)!!\right)+I_{q}(2[q/2])+o(1),
\nonumber
\end{align}
as $\rho\rightarrow\infty$. Recall that $I_{q}(2[q/2])=\sup_{\theta\in\mathbb{R}}\{\theta 2[q/2]-f_{q}(\theta)\}$ and the optimal
$\theta$ satisfies
\begin{equation}
f'_{q}(\theta)=\frac{\sum_{j=0}^{[q/2]}\frac{q!}{(2j)!(q-2j)!}(2j-1)!!e^{2j\theta}(2j)}
{\sum_{j=0}^{[q/2]}\frac{q!}{(2j)!(q-2j)!}(2j-1)!!e^{2j\theta}}=2[q/2],
\end{equation}
which implies that $\theta=\infty$. But as $\theta\rightarrow\infty$,
\begin{align}
&\theta 2[q/2]-\log\left(\sum_{j=0}^{[q/2]}\frac{q!}{(2j)!(q-2j)!}(2j-1)!!e^{2j\theta}\right)
\\
&\qquad\qquad
\rightarrow-\log\left(\frac{q!}{(2[q/2])!(q-2[q/2])!}(2[q/2]-1)!!\right).
\nonumber
\end{align}
Hence, from \eqref{DifferenceBound}, we get
\begin{equation}
\sup_{0\leq x\leq 2[q/2]}\left\{x\log\rho-I_{q}(x)\right\}-\{2[q/2]\log\rho-I_{q}(2[q/2])\}\rightarrow 0,
\end{equation}
as $\rho\rightarrow\infty$. Therefore, \eqref{LargeRho} follows from (ii).
\end{proof}

\begin{proof}[Proof of Proposition \ref{prop:LambdaSimple}]
The functional $\Lambda[h]$ can be put in the following form
by expressing the double integral using the integral equation 
(\ref{EulerLagrange}) 
\begin{equation}\label{Lam1}
\Lambda[f]  = \frac12 \log\rho \int_0^1  f(z) dz+ \frac12
\int_0^1 f(x) I'_q(f(x)) dx- \int_0^1 I_q(f(x)) dx
\end{equation}
We will show next that the integrals appearing in this formula
can be expressed only in terms of $h(1)$. 

The first integral is
\begin{equation}\label{I0def}
I_0 \equiv \int_0^1 f(x) dx= \frac{1}{2\beta} h'(0) = \frac{1}{\sqrt{\beta}}
\sqrt{f_q(h(1)) - f_q(\log \rho)}
\end{equation}

The second integral can be expressed using the relation (\ref{fvsh2}) 
\begin{equation}
I_1 \equiv \int_0^1 f(x) I'_q(f(x)) dx= \int_0^1 f(x) h(x) dx= 
\frac{1}{2\beta} h'(0) h(0) + \frac{1}{2\beta}
\int_0^1 [h'(y)]^2 dy
\end{equation}

Finally, the last integral can be computed using the 
explicit form of the Legendre transform
$I_q(f(x)) = f(x) h(x) - f_q(h(x))$ as
\begin{align}
I_2 \equiv \int_0^1 I_q(f(x)) dx
&= \int_0^1 \{ f(x) h(x) - f_q(h(x)) \} dx
\\
&= I_1 - \frac{1}{2\beta}
\int_0^1 V(h(x)) dx
\nonumber
\\
&= 
I_1 - \frac{1}{2\beta} V(h(1)) + \frac{1}{4\beta} \int_0^1 [h'(y)]^2 dy \,.
\nonumber
\end{align}

Both $I_1$ and $I_2$ contain the integral of $[h'(y)]^2$ which would appear
to require knowledge of the entire function $h(y)$. However, this integral
can be simplified if we change the integration variable to $h(y)$. This gives
\begin{align}\label{I3}
\int_{0}^{1}[h'(y)]^{2}dy
&=\int_{h(0)=\log\rho}^{h(1)} h'(y) dh(y)
\\
&=\int_{h(0)=\log\rho}^{h(1)} \sqrt{2[V(h(1)) - V(x)]} dx
\nonumber
\\
&= 2\sqrt{\beta} \int_{\log\rho}^{h(1)} \sqrt{f_q(h(1)) - f_q(x)} dx. 
\nonumber
\end{align}

We can insert now all the pieces into (\ref{Lam1}) and get
\begin{align}
\Lambda[f]&=\frac12\log\rho I_0 + \frac12 I_1 - I_2
\\
&= \frac12\log\rho I_0 - \frac12 I_1 + \frac{1}{2\beta} V(h(1)) - \frac{1}{4\beta}
 \int_0^1 [h'(y)]^2 dy
\nonumber
\\
&=\frac12\log\rho I_0 - \frac{1}{4\beta} h'(0) h(0) + \frac{1}{2\beta} V(h(1))
- \frac{1}{2\beta}
 \int_0^1 [h'(y)]^2 dy
\nonumber
\end{align}
The first two terms cancel each other. Using the result (\ref{I3}) for the
integral in the last term we get the final result (\ref{Lambdasimple}) for 
the functional $\Lambda[f]$.
This concludes the proof of this result. 
\end{proof}

\begin{proof}[Proof of Proposition \ref{ASIThm}]
In the Euler-Log-Euler discretization, we have
\begin{equation}
|S_{n}|=S_{0}e^{\sum_{i=0}^{n-1}\log|1+\sigma_{i}\sqrt{\tau}\varepsilon_{i}|}.
\end{equation}
First, we claim that
\begin{equation}\label{ASI}
\lim_{n\rightarrow\infty}\left|\frac{1}{n}\sum_{i=0}^{n-1}\log|1+\sigma_{i}\sqrt{\tau}\varepsilon_{i}|
-\frac{1}{n}\sum_{i=0}^{n-1}\log|1+\sigma_{0}\sqrt{\tau}\varepsilon_{i}|\right|
= 0 \qquad\text{a.s.}
\end{equation}
We start by proving 
that $\sigma_{n}\rightarrow\sigma_{0}$ a.s. as $n\rightarrow\infty$.
Notice that $\frac{1}{2}\omega^{2}t_{n}=\frac{1}{2}\omega^{2}n\tau=\frac{\beta}{n}\rightarrow 0$
as $n\rightarrow\infty$ and $\omega Z_{n}=\omega\sqrt{\tau}B_{n}=\frac{\sqrt{2\beta}}{n}B_{n}$, 
where $B_{n}$ is a standard Brownian motion.
It is well known that for standard Brownian motion $B_{t}$, we have 
$|B_{t}|/t\rightarrow 0$ a.s. as $t\rightarrow\infty$.
Thus, we conclude that $\sigma_{n}\rightarrow\sigma_{0}$ a.s. as $n\rightarrow\infty$.

From this it follows that for any $\epsilon>0$, for a.e. $\omega$, there 
exists $N(\omega)$ so that for any $n\geq N$, 
$|\sigma_{n}-\sigma_{0}|\leq\epsilon$. It is clear that
\begin{equation}
\lim_{n\rightarrow\infty}\frac{1}{n}\sum_{i=0}^{N-1}\log|1+\sigma_{i}\sqrt{\tau}\varepsilon_{i}|
=\lim_{n\rightarrow\infty}\frac{1}{n}\sum_{i=0}^{N-1}\log|1+\sigma_{0}\sqrt{\tau}\varepsilon_{i}|=0.
\end{equation}
On the other hand,
\begin{align}
&\left|\frac{1}{n}\sum_{i=N}^{n-1}\log|1+\sigma_{i}\sqrt{\tau}\varepsilon_{i}|
-\frac{1}{n}\sum_{i=N}^{n-1}\log|1+\sigma_{0}\sqrt{\tau}\varepsilon_{i}|\right|
\\
&\leq\frac{1}{n}\sum_{i=N}^{n-1}\left|\log\left|\frac{1+\sigma_{i}\sqrt{\tau}\varepsilon_{i}}
{1+\sigma_{0}\sqrt{\tau}\varepsilon_{i}}\right|\right|
\nonumber
\\
&=\frac{1}{n}\sum_{i=N}^{n-1}\left|\log\left|1+\frac{(\sigma_{i}-\sigma_{0})\sqrt{\tau}\varepsilon_{i}}
{1+\sigma_{0}\sqrt{\tau}\varepsilon_{i}}\right|\right|
\nonumber
\\
&\leq\frac{1}{n}\sum_{i=N}^{n-1}\log\left(1+\left|\frac{(\sigma_{i}-\sigma_{0})\sqrt{\tau}\varepsilon_{i}}
{1+\sigma_{0}\sqrt{\tau}\varepsilon_{i}}\right|\right)
\nonumber
\\
&\leq\frac{1}{n}\sum_{i=N}^{n-1}\log\left(1+\epsilon\left|\frac{|\sqrt{\tau}\varepsilon_{i}|}
{1+\sigma_{0}\sqrt{\tau}\varepsilon_{i}}\right|\right).
\nonumber
\end{align}
One can check that $\mathbb{E}[\log(1+\epsilon|\frac{\sqrt{\tau}\varepsilon_{1}}
{1+\sigma_{0}\sqrt{\tau}\varepsilon_{1}}|)]$ is finite. By strong law of large numbers,
\begin{align}
&\limsup_{n\rightarrow\infty}\left|\frac{1}{n}\sum_{i=N}^{n-1}\log|1+\sigma_{i}\sqrt{\tau}\varepsilon_{i}|
-\frac{1}{n}\sum_{i=N}^{n-1}\log|1+\sigma_{0}\sqrt{\tau}\varepsilon_{i}|\right|
\\
&\qquad\qquad\qquad
\leq\mathbb{E}\left[\log\left(1+\epsilon\left|\frac{\sqrt{\tau}\varepsilon_{1}}
{1+\sigma_{0}\sqrt{\tau}\varepsilon_{1}}\right|\right)\right]\qquad\text{a.s.}
\nonumber
\end{align}
Let $\epsilon\rightarrow 0$ and by monotone convergence theorem,
\begin{equation}
\lim_{n\rightarrow\infty}\left|\frac{1}{n}\sum_{i=N}^{n-1}\log|1+\sigma_{i}\sqrt{\tau}\varepsilon_{i}|
-\frac{1}{n}\sum_{i=N}^{n-1}\log|1+\sigma_{0}\sqrt{\tau}\varepsilon_{i}|\right|=0.\qquad\text{a.s.}
\end{equation}
This concludes the proof of (\ref{ASI}).

Finally, since $\mathbb{E}[\log|1+\rho\varepsilon_{1}|]$ is finite, 
by the strong law of large numbers,
\begin{equation}\label{ASII}
\lim_{n\rightarrow\infty}
\frac{1}{n}\sum_{i=0}^{n-1}\log|1+\sigma_{0}\sqrt{\tau}\varepsilon_{i}|
=\mathbb{E}\left[\log\left|1+\rho\varepsilon_{1}\right|\right]\qquad\text{a.s.}
\end{equation}
By \eqref{ASI} and \eqref{ASII}, we conclude that
\begin{equation}
\lim_{n\rightarrow\infty}\frac{1}{n}\log|S_{n}|=\mathbb{E}\left[\log\left|1+\rho\varepsilon_{1}\right|\right]\qquad\text{a.s.}
\end{equation}
\end{proof}

\begin{proof}[Proof of Proposition \ref{CLTII}]
The proof is similar to the proof of Proposition \ref{CLTI}. We will give an outline here and detailed
estimates are omitted. Let $\mathbb{E}_{V}$ and $\text{Var}_{V}$ be the expectation and variance conditional on $(V_{j})_{j=1}^{\infty}$.
For any $\theta\in\mathbb{R}$,
\begin{align}
&\mathbb{E}\left[e^{\frac{i\theta}{\sqrt{n}}(\log|S_{n}|-\mathbb{E}[\log|1+\rho\varepsilon_{1}|]n)}\right]
\\
&=S_{0}^{\frac{i\theta}{\sqrt{n}}}\mathbb{E}\left[e^{\sum_{i=0}^{n-1}\log\mathbb{E}_{V}\left[e^{\frac{i\theta}{\sqrt{n}}\log|1+\sigma_{i}\varepsilon_{1}|}\right]
-\sqrt{n}i\theta\mathbb{E}[\log|1+\rho\varepsilon_{1}|]}\right]
\nonumber
\\
&=S_{0}^{\frac{i\theta}{\sqrt{n}}}\mathbb{E}\left[e^{\sum_{i=0}^{n-1}\log(1+\frac{i\theta}{\sqrt{n}}\mathbb{E}_{V}\log|1+\sigma_{i}\varepsilon_{1}|
-\frac{\theta^{2}}{2n}\mathbb{E}_{V}[(\log|1+\sigma_{i}\varepsilon_{1}|)^{2}]+O(n^{-3/2}))-\sqrt{n}i\theta\mathbb{E}[\log|1+\rho\varepsilon_{1}|]}\right]
\nonumber
\\
&=S_{0}^{\frac{i\theta}{\sqrt{n}}}\mathbb{E}\left[e^{\frac{i\theta}{\sqrt{n}}
\sum_{i=0}^{n-1}(\mathbb{E}_{V}\log|1+\sigma_{i}\varepsilon_{1}|-\mathbb{E}[\log|1+\rho\varepsilon_{1}|])
-\frac{\theta^{2}}{2n}\sum_{i=0}^{n-1}\text{Var}_{V}[\log|1+\sigma_{i}\varepsilon_{1}|]+O(n^{-3/2})}\right]
\nonumber
\\
&=S_{0}^{\frac{i\theta}{\sqrt{n}}}\mathbb{E}\left[e^{\frac{i\theta}{\sqrt{n}}
\sum_{i=0}^{n-1}\mathbb{E}_{V}\log\left|1+\frac{\rho(\exp\{\sum_{j=1}^{i}\omega\sqrt{\tau}V_{j}\}-1)\varepsilon_{1}}{1+\rho\varepsilon_{1}}\right|
-\frac{\theta^{2}}{2n}\sum_{i=0}^{n-1}\text{Var}_{V}[\log|1+\sigma_{i}\varepsilon_{1}|]+O(n^{-3/2})}\right]
\nonumber
\\
&=S_{0}^{\frac{i\theta}{\sqrt{n}}}\mathbb{E}\left[e^{\frac{i\theta}{\sqrt{n}}
\sum_{i=0}^{n-1}\mathbb{E}_{V}\log\left|1+\frac{\rho(\exp\{\sum_{j=1}^{i}\frac{\sqrt{2\beta}}{n}V_{j}\}-1)\varepsilon_{1}}{1+\rho\varepsilon_{1}}\right|
-\frac{\theta^{2}}{2n}\sum_{i=0}^{n-1}\text{Var}_{V}[\log|1+\sigma_{i}\varepsilon_{1}|]+O(n^{-3/2})}\right]
\nonumber
\\
&=S_{0}^{\frac{i\theta}{\sqrt{n}}}\mathbb{E}\left[e^{\frac{i\theta}{\sqrt{n}}
H'(0)\frac{\rho\sqrt{2\beta}}{n}\sum_{i=0}^{n-1}\sum_{j=1}^{i}V_{j}
-\frac{\theta^{2}}{2n}\sum_{i=0}^{n-1}\text{Var}_{V}[\log|1+\sigma_{i}\varepsilon_{1}|]+O(n^{-3/2})}\right]
\nonumber
\\
&\rightarrow e^{-\frac{\theta^{2}}{3}(H'(0))^{2}\rho^{2}\beta
-\frac{\theta^{2}}{2}\text{Var}[\log|1+\rho\varepsilon_{1}|]},
\nonumber
\end{align}
as $n\rightarrow\infty$, where $H(x):=\mathbb{E}[\log|1+\frac{\varepsilon_{1}x}{1+\rho\varepsilon_{1}}|]$.
\end{proof}

\subsection{Proofs of the Results in Section \ref{LogEulerLogEulerSection}}

\begin{proof}[Proof of Theorem \ref{LogEulerLogEulerThm}]
Since $\mathbb{E}[e^{\theta X}]=\infty$ for any $\theta>0$ for any log-normal random variable $X$, 
it is clear that $\mathbb{E}[(S_{n})^{q}]=\infty$ for $q>1$ or $q<0$. 
For $0\leq q\leq 1$,
\begin{align}
\mathbb{E}[(S_{n})^{q}]
&=S_{0}^{q}\mathbb{E}\left[e^{\frac{\rho^{2}q(q-1)}{2}\sum_{k=0}^{n-1}
e^{2\omega Z_{k}-\omega^{2}t_{k}}}\right]
\\
&=S_{0}^{q}\mathbb{E}\left[e^{\frac{\rho^{2}q(q-1)}{2}\sum_{k=0}^{n-1}
e^{2\omega\sqrt{\tau}X_{k}-\omega^{2}k\tau}}\right]
\nonumber
\\
&=S_{0}^{q}\mathbb{E}\left[e^{\frac{\rho^{2}q(q-1)}{2}\sum_{k=0}^{n-1}
e^{2\omega\sqrt{\tau}\sum_{j=1}^{k}V_{j}-\omega^{2}k\tau}}\right]
\nonumber
\\
&=S_{0}^{q}\mathbb{E}\left[e^{\frac{\rho^{2}q(q-1)}{2}\sum_{k=0}^{n-1}
e^{\frac{2\sqrt{2\beta}}{n}\sum_{j=1}^{k}V_{j}-\frac{2\beta}{n^{2}}k}}\right],
\nonumber
\end{align}
where $X_{k}:=\frac{Z_{k}}{\sqrt{\tau}}$ and $V_{j}:=X_{j}-X_{j-1}$, $1\leq j\leq k$, are i.i.d. $N(0,1)$ random variables.
Note that $\sum_{j=1}^{0}V_{j}$ is defined as $0$ to be consistent with $X_{0}=0$.
By Mogulskii theorem, $\mathbb{P}(\frac{1}{n}\sum_{j=1}^{\lfloor\cdot n\rfloor}V_{j}\in\cdot)$
satisfies a large deviation principle on $L^{\infty}[0,1]$ with rate function
\begin{equation}
I(g)=\frac{1}{2}\int_{0}^{1}(g'(x))^{2}dx,
\end{equation}
if $g\in\mathcal{AC}_{0}[0,1]$, i.e., absolutely continuous and $g(0)=0$ and $I(g)=+\infty$ otherwise.
Let $g(x):=\frac{1}{n}\sum_{j=1}^{\lfloor xn\rfloor}V_{j}$. Then, 
\begin{equation}
\int_{0}^{1}e^{2\sqrt{2\beta}g(x)}dx
=\sum_{k=0}^{n-1}\int_{\frac{k}{n}}^{\frac{k+1}{n}}e^{2\sqrt{2\beta}g(x)}dx
=\frac{1}{n}\sum_{k=0}^{n-1}e^{2\sqrt{2\beta}\sum_{j=1}^{k}V_{j}}.
\end{equation}
Moreover, we claim that
\begin{equation}
g\mapsto\int_{0}^{1}e^{2\sqrt{2\beta}g(x)}dx
\end{equation}
is a bounded and continuous map. Since $g\in L^{\infty}[0,1]$, clearly it is a bounded map.
Now assume that $g_{n}\rightarrow g$ in $L^{\infty}[0,1]$. Observe that for any $|x|\leq\frac{1}{2}$.
\begin{align}
|e^{x}-1|&=\left|x+\frac{x^{2}}{2!}+\frac{x^{3}}{3!}+\cdots\right|
\\
&\leq |x|(1+|x|+|x|^{2}+\cdots)
\nonumber
\\
&\leq 2|x|.
\nonumber
\end{align}
Let $n$ be sufficiently large so that $2\sqrt{2\beta}\Vert g_{n}-g\Vert_{L^{\infty}[0,1]}\leq\frac{1}{2}$. 
Therefore, we have
\begin{align}
&\left|\int_{0}^{1}e^{2\sqrt{2\beta}g_{n}(x)}dx
-\int_{0}^{1}e^{2\sqrt{2\beta}g(x)}dx\right|
\\
&=\left|\int_{0}^{1}e^{2\sqrt{2\beta}g(x)}\left(e^{2\sqrt{2\beta}(g_{n}(x)-g(x))}
-1\right)dx\right|
\nonumber
\\
&\leq e^{2\sqrt{2\beta}\Vert g\Vert_{L^{\infty}[0,1]}}
\int_{0}^{1}\left|e^{2\sqrt{2\beta}(g_{n}(x)-g(x))}
-1\right|dx
\nonumber
\\
&\leq 4\sqrt{2\beta}e^{2\sqrt{2\beta}\Vert g\Vert_{L^{\infty}[0,1]}}\Vert g_{n}-g\Vert_{L^{\infty}[0,1]}
\nonumber
\end{align}
which converges to $0$ as $n\rightarrow\infty$. Hence the map is continuous. 
By Varadhan's lemma,
\begin{align}
&\lim_{n\rightarrow\infty}\frac{1}{n}\log\mathbb{E}\left[e^{\frac{\rho^{2}q(q-1)}{2}\sum_{k=0}^{n-1}
e^{\frac{2\sqrt{2\beta}}{n}\sum_{j=1}^{k}V_{j}}}\right]
\\
&=\sup_{g\in\mathcal{AC}_{0}[0,1]}\left\{\frac{\rho^{2}q(q-1)}{2}\int_{0}^{1}e^{2\sqrt{2\beta}g(x)}dx
-\frac{1}{2}\int_{0}^{1}(g'(x))^{2}dx\right\}.
\nonumber
\end{align}
Finally, notice that
\begin{align}
S_{0}^{q}\mathbb{E}\left[e^{\frac{\rho^{2}q(q-1)}{2}\sum_{k=0}^{n-1}
e^{\frac{2\sqrt{2\beta}}{n}\sum_{j=1}^{k}V_{j}}}\right]&\leq\mathbb{E}[S_{n}^{q}]
\\
&\leq S_{0}^{q}\mathbb{E}\left[e^{\frac{\rho^{2}q(q-1)}{2}e^{-\frac{2\beta}{n}}\sum_{k=0}^{n-1}
e^{\frac{2\sqrt{2\beta}}{n}\sum_{j=1}^{k}V_{j}}}\right].
\nonumber
\end{align}
Hence, for any $0\leq q\leq 1$,
\begin{align}
&\lim_{n\rightarrow\infty}\frac{1}{n}\log\mathbb{E}[(S_{n})^{q}]
\\
&=\sup_{g\in\mathcal{AC}_{0}[0,1]}\left\{\frac{\rho^{2}q(q-1)}{2}\int_{0}^{1}e^{2\sqrt{2\beta}g(x)}dx
-\frac{1}{2}\int_{0}^{1}(g'(x))^{2}dx\right\}.
\nonumber
\end{align}
Indeed, for $q\notin[0,1]$, if we let $g(x)=Kx$ and let $K\rightarrow\infty$, 
we have $\lim_{n\rightarrow\infty}\frac{1}{n}\log\mathbb{E}[(S_{n})^{q}]=\infty$, which is
consistent with the discussions before.
\end{proof}

\begin{proof}[Proof of Proposition \ref{prop:lamab}]

The solution of the variational problem appearing in Proposition~\ref{prop:lamab}
can be extracted from the Corollary 5 in \cite{PZAsian}. 
We sketch here the main steps for completeness of the presentation.

At optimality the function $g(x)$ satisfies the Euler-Lagrange equation
\begin{equation}\label{ELeq}
g''(x) = ab e^{bg(x)}
\end{equation}
with boundary conditions
\begin{equation}
g(0) = 0 \,,\quad g'(1) = 0\,.
\end{equation}
The condition at $x=1$ is a transversality condition.

The solution of the equation (\ref{ELeq}) with $a >0, b>0$ is
\begin{equation}\label{gsol}
g(x) = \frac{1}{b} \log\left( \frac{\cos^2\xi}{\cos^2(\xi(x-1))} \right)
\end{equation}
where $\xi \in (0,\frac{\pi}{2})$ is the solution of the equation
\begin{equation}
2\xi^2 = ab^2 \cos^2\xi\,.
\end{equation}
Substituting the solution (\ref{gsol}) into the functional of 
Proposition~\ref{prop:lamab} and performing the integrations gives the
result (\ref{lamabsol}) for $\lambda(a,b)$. 

\end{proof}

\begin{proof}[Proof of Proposition~\ref{ASIIProp}]
\begin{equation}
S_{n}=S_{0}e^{\sum_{i=0}^{n-1}\sigma_{i}\Delta W_{i}-\frac{1}{2}\sigma_{i}^{2}\tau},
\end{equation}
where $\sigma_{i}=\sigma_{0}e^{\omega Z_{i}-\frac{1}{2}\omega^{2}t_{i}}$.
We showed in the proof of Proposition~\ref{ASIThm} that $\lim_{n\to \infty}
\sigma_n = \sigma_0$.
Next, notice that
\begin{equation}
S_{n}=S_{0}e^{\sum_{i=0}^{n-1}\sigma_{i}\sqrt{\tau}\Delta B_{i}-\frac{1}{2}\sigma_{i}^{2}\tau},
\end{equation}
where $B_{t}$ is a standard Brownian motion and
\begin{equation}
\lim_{n\rightarrow\infty}\frac{1}{n}\sum_{i=0}^{n-1}\sigma_{i}\sqrt{\tau}\Delta B_{i}
=\sigma_{0}\sqrt{\tau}\lim_{n\rightarrow\infty}\frac{B_{n}}{n}=0,\qquad\text{a.s.}
\end{equation}
On the other hand,
$\lim_{n\rightarrow\infty}\frac{1}{n}\sum_{i=0}^{n-1}\frac{1}{2}\sigma_{i}^{2}\tau=\frac{1}{2}\sigma_{0}^{2}\tau
=\frac{1}{2}\rho^{2}$ a.s. 
Therefore, we conclude that $S_{n}\rightarrow 0$ a.s. as $n\rightarrow\infty$. More precisely,
\begin{equation}
\lim_{n\rightarrow\infty}\frac{1}{n}\log S_{n}=-\frac{1}{2}\rho^{2},\qquad\text{a.s.}
\end{equation}
\end{proof}

\begin{proof}[Proof of Proposition \ref{CLTI}]
We can compute that, for any $\theta\in\mathbb{R}$,
\begin{align}
\mathbb{E}\left[e^{i\theta\frac{\log S_{n}+\frac{1}{2}\rho^{2}n}{\sqrt{n}}}\right]
&=\mathbb{E}\left[S_{n}^{\frac{i\theta}{\sqrt{n}}}\right]e^{i\theta\sqrt{n}\frac{1}{2}\rho^{2}}
\\
&=S_{0}^{\frac{i\theta}{\sqrt{n}}}\mathbb{E}\left[e^{\frac{1}{2}(\frac{i\theta}{\sqrt{n}})^{2}\sum_{i=0}^{n-1}\sigma_{i}^{2}\tau}
e^{-\frac{1}{2}\frac{i\theta}{\sqrt{n}}\sum_{i=0}^{n-1}(\sigma_{i}^{2}\tau-\rho^{2})}\right].
\nonumber
\end{align}
Note that $\sigma_{n}=\sigma_{0}e^{\omega Z_{n}-\frac{1}{2}\omega^{2}t_{n}}$.
We have $\frac{1}{2}\omega^{2}t_{n}=\frac{1}{2}\omega^{2}n\tau=\frac{\beta}{n}\rightarrow 0$
as $n\rightarrow\infty$ and $\omega Z_{n}=\omega\sqrt{\tau}B_{n}=\frac{\sqrt{2\beta}}{n}B_{n}$, 
where $B_{n}$ is a standard Brownian motion. Thus, $\sigma_{n}\rightarrow\sigma_{0}$ a.s. as $n\rightarrow\infty$
and since $\rho=\sigma_{0}\sqrt{\tau}$, we have
\begin{equation}
\frac{1}{2}\left(\frac{i\theta}{\sqrt{n}}\right)^{2}\sum_{i=0}^{n-1}\sigma_{i}^{2}\tau\rightarrow-\frac{1}{2}\theta^{2}\rho^{2},
\end{equation}
a.s. as $n\rightarrow\infty$. Moreover,
\begin{align}\label{DiffTerm}
\frac{1}{2}\frac{1}{\sqrt{n}}\sum_{i=0}^{n-1}(\sigma_{i}^{2}\tau-\rho^{2})
&=\frac{1}{2}\frac{1}{\sqrt{n}}\rho^{2}\sum_{i=0}^{n-1}(e^{2\omega Z_{i}-\omega^{2}t_{i}}-1)
\\
&=\frac{1}{2}\frac{1}{\sqrt{n}}\rho^{2}\sum_{i=0}^{n-1}(e^{\frac{2\sqrt{2\beta}}{n}B_{i}-2\frac{\beta i}{n^{2}}}-1)
\nonumber
\\
&=\frac{1}{2}\frac{1}{\sqrt{n}}\rho^{2}\sum_{i=0}^{n-1}\frac{2\sqrt{2\beta}}{n}B_{i}
+\frac{1}{2}\xi_{n},
\nonumber
\end{align}
where
\begin{equation}
\xi_{n}:=\frac{1}{\sqrt{n}}\rho^{2}\sum_{i=0}^{n-1}\left[e^{\frac{2\sqrt{2\beta}}{n}B_{i}-2\frac{\beta i}{n^{2}}}-1-\frac{2\sqrt{2\beta}}{n}B_{i}\right].
\end{equation}
First, we claim that $\xi_{n}\rightarrow 0$ in probability as $n\rightarrow\infty$.
On the one hand,
\begin{align}
\xi_{n}
&\leq\frac{1}{\sqrt{n}}\rho^{2}\sum_{i=0}^{n-1}\left[e^{\frac{2\sqrt{2\beta}}{n}B_{i}}-1-\frac{2\sqrt{2\beta}}{n}B_{i}\right]
\\
&=\frac{1}{\sqrt{n}}\rho^{2}\sum_{i=0}^{n-1}\sum_{k=2}^{\infty}\frac{1}{k!}\frac{(2\sqrt{2\beta})^{k}}{n^{k}}B_{i}^{k}
\nonumber
\\
&\leq\frac{1}{\sqrt{n}}\rho^{2}\sum_{i=0}^{n-1}\sum_{k=2}^{\infty}\frac{1}{k!}\frac{(2\sqrt{2\beta})^{k}}{n^{k}}|B_{i}|^{k}=:\overline{\xi}_{n},
\nonumber
\end{align}
and it is easy to check that $\mathbb{E}[\overline{\xi}_{n}]=O(\frac{1}{\sqrt{n}})$. On the other hand, 
since $e^{-x}\geq 1-x$ for all real $x$, 
\begin{align}
\xi_{n}
&\geq\frac{1}{\sqrt{n}}\rho^{2}\sum_{i=0}^{n-1}\left[e^{\frac{2\sqrt{2\beta}}{n}B_{i}-2\frac{\beta}{n}}-1-\frac{2\sqrt{2\beta}}{n}B_{i}\right]
\\
&\geq
\frac{1}{\sqrt{n}}\rho^{2}\sum_{i=0}^{n-1}\left[e^{\frac{2\sqrt{2\beta}}{n}B_{i}}-1-\frac{2\sqrt{2\beta}}{n}B_{i}\right]
-\frac{2\beta}{n}\frac{1}{\sqrt{n}}\rho^{2}\sum_{i=0}^{n-1}e^{\frac{2\sqrt{2\beta}}{n}B_{i}}
\nonumber
\\
&\geq-\frac{2\beta}{n}\frac{1}{\sqrt{n}}\rho^{2}\sum_{i=0}^{n-1}e^{\frac{2\sqrt{2\beta}}{n}|B_{i}|}
-\overline{\xi}_{n}=:-\underline{\xi}_{n},
\nonumber
\end{align}
and it is easy to check that $\mathbb{E}[\underline{\xi}_{n}]=O(\frac{1}{\sqrt{n}})$.
Hence, we proved that $\xi_{n}\rightarrow 0$ in probability as $n\rightarrow\infty$.
Next, we turn to the first term in the last line of \eqref{DiffTerm}. We claim that
\begin{equation}
\frac{1}{2}\frac{1}{\sqrt{n}}\rho^{2}\sum_{i=0}^{n-1}\frac{2\sqrt{2\beta}}{n}B_{i}
\rightarrow N\left(0,\frac{2}{3}\rho^{4}\beta\right),
\end{equation}
in distribution as $n\rightarrow\infty$. To show this, note that we can write $B_{i}=\sum_{j=1}^{i}V_{j}$ ($\sum_{j=1}^{0}V_{j}:=0$)
for i.i.d. $N(0,1)$ random variables $(V_{j})_{j\in\mathbb{N}\cup\{0\}}$. 
Thus, for any $\theta\in\mathbb{R}$, 
\begin{align}
\mathbb{E}\left[e^{\frac{i\theta}{2\sqrt{n}}\rho^{2}\sum_{i=0}^{n-1}\frac{2\sqrt{2\beta}}{n}B_{i}}\right]
&=\mathbb{E}\left[e^{\frac{i\theta}{2\sqrt{n}}\rho^{2}\sum_{i=0}^{n-1}\frac{2\sqrt{2\beta}}{n}\sum_{j=1}^{i}V_{j}}\right]
\\
&=\mathbb{E}\left[e^{\frac{i\theta}{2\sqrt{n}}\rho^{2}\frac{2\sqrt{2\beta}}{n}[(n-1)V_{1}+(n-2)V_{2}+\cdots+1\cdot V_{n-1}]}\right]
\nonumber
\\
&=e^{-\frac{\theta^{2}\rho^{4}\beta}{n^{3}}[(n-1)^{2}+(n-2)^{2}+\cdots+1^{2}]}
\nonumber
\\
&=e^{-\frac{\theta^{2}\rho^{4}\beta}{n^{3}}\frac{1}{6}(n-1)n(2n-1)}
\nonumber
\\
&\rightarrow e^{-\frac{1}{3}\theta^{2}\rho^{4}\beta},
\nonumber
\end{align}
as $n\rightarrow\infty$. Putting everything together, we have
\begin{equation}
\mathbb{E}\left[e^{i\theta\frac{\log S_{n}+\frac{1}{2}\rho^{2}n}{\sqrt{n}}}\right]
\rightarrow e^{-\frac{1}{2}\theta^{2}\rho^{2}-\frac{1}{3}\theta^{2}\rho^{4}\beta},
\end{equation}
as $n\rightarrow\infty$ and therefore we proved the desired result.
\end{proof}

\subsection{Proofs of the Results in Section \ref{LogEulerEulerSection}}

\begin{proof}[Proof of Theorem \ref{LogEulerEulerThm}]
When $q<0$ or $q>1$, 
\begin{align}
\mathbb{E}[(S_{n})^{q}]
&=S_{0}^{q}\mathbb{E}\left[e^{\frac{\rho^{2}}{2}q(q-1)\sum_{k=0}^{n-1}\prod_{j=1}^{k}(1+\frac{\sqrt{2\beta}}{n}V_{j})^{2}}\right]
\\
&\geq
S_{0}^{q}\mathbb{E}\left[e^{\frac{\rho^{2}}{2}q(q-1)\sum_{k=0}^{n-1}\prod_{j=1}^{k}(1+\frac{\sqrt{2\beta}}{n}V_{j})^{2}}
1_{V_{j}\geq\sqrt{n},1\leq j\leq n-1}\right]
\nonumber
\\
&\geq
S_{0}^{q}e^{\frac{\rho^{2}}{2}q(q-1)(1+\frac{\sqrt{2\beta}}{\sqrt{n}})^{(n-1)n}}
\mathbb{P}(V_{1}\geq\sqrt{n})^{n-1}.
\nonumber
\end{align}
$\mathbb{P}(V_{1}\geq\sqrt{n})^{n-1}$ behaves like $e^{-n(n-1)/2}$
and $(1+\frac{\sqrt{2\beta}}{\sqrt{n}})^{(n-1)n}$ behaves like $e^{\sqrt{2\beta}(n-1)}$ and thus
\begin{equation}
\lim_{n\rightarrow\infty}\frac{1}{n}\log\mathbb{E}[(S_{n})^{q}]=\infty.
\end{equation}

Now assume that $0\leq q\leq 1$.
\begin{align}
\mathbb{E}[(S_{n})^{q}]
&=S_{0}^{q}\mathbb{E}\left[e^{\frac{1}{2}q(q-1)\sum_{k=0}^{n-1}\sigma_{k}^{2}\tau}\right]
\\
&=S_{0}^{q}\mathbb{E}\left[e^{\frac{\rho^{2}}{2}q(q-1)\sum_{k=0}^{n-1}\prod_{j=1}^{k}(1+\omega\sqrt{\tau}V_{j})^{2}}\right]
\nonumber
\\
&=S_{0}^{q}\mathbb{E}\left[e^{\frac{\rho^{2}}{2}q(q-1)\sum_{k=0}^{n-1}e^{2\sum_{j=1}^{k}\log|1+\frac{\sqrt{2\beta}}{n}V_{j}|}}\right]
\nonumber
\\
&=S_{0}^{q}\mathbb{E}\left[e^{\frac{\rho^{2}}{2}q(q-1)\sum_{k=0}^{n-1}e^{\frac{2\sqrt{2\beta}}{n}\sum_{j=1}^{k}V_{j}^{(n)}}}\right],
\nonumber
\end{align}
where
\begin{equation}
V_{j}^{(n)}:=n\frac{1}{\sqrt{2\beta}}\log\left|1+\frac{\sqrt{2\beta}}{n}V_{j}\right|.
\end{equation}

By Lemma \ref{EquivLemma}, $\frac{1}{n}\sum_{j=1}^{\lfloor\cdot n\rfloor}V_{j}^{(n)}$
and $\frac{1}{n}\sum_{j=1}^{\lfloor\cdot n\rfloor}V_{j}$ are exponentially equivalent, see e.g. Dembo and Zeitouni \cite{Dembo}.
By Mogulskii theorem, $\mathbb{P}(\frac{1}{n}\sum_{j=1}^{\lfloor\cdot n\rfloor}V_{j}\in\cdot)$
satisfies a large deviation principle on $L^{\infty}[0,1]$ with rate function
\begin{equation}
I(g)=\frac{1}{2}\int_{0}^{1}(g'(x))^{2}dx,
\end{equation}
if $g\in\mathcal{AC}_{0}[0,1]$, i.e., absolutely continuous and $g(0)=0$ and $I(g)=+\infty$ otherwise.
Since $\frac{1}{n}\sum_{j=1}^{\lfloor\cdot n\rfloor}V_{j}^{(n)}$
and $\frac{1}{n}\sum_{j=1}^{\lfloor\cdot n\rfloor}V_{j}$ are exponentially equivalent,
$\mathbb{P}(\frac{1}{n}\sum_{j=1}^{\lfloor\cdot n\rfloor}V_{j}^{(n)}\in\cdot)$ also satisfies a large deviation principle
with the same rate function $I(g)$.

Following the same arguments as in the log-Euler-log-Euler scheme, we conclude that for $0\leq q\leq 1$,
\begin{align}
&\lim_{n\rightarrow\infty}\frac{1}{n}\log\mathbb{E}[(S_{n})^{q}]
\\
&=\sup_{g\in\mathcal{AC}_{0}[0,1]}\left\{\frac{\rho^{2}q(q-1)}{2}\int_{0}^{1}e^{2\sqrt{2\beta}g(x)}dx
-\frac{1}{2}\int_{0}^{1}(g'(x))^{2}dx\right\}.
\nonumber
\end{align}
\end{proof}

\begin{proof}[Proof of Lemma \ref{EquivLemma}]
First, we observe that
\begin{align}
&\mathbb{P}\left(\sup_{0\leq x\leq 1}\left|\frac{1}{n}\sum_{j=1}^{\lfloor xn\rfloor}V_{j}^{(n)}-\frac{1}{n}\sum_{j=1}^{\lfloor xn\rfloor}V_{j}\right|
\geq\epsilon\right)
\\
&\leq
\mathbb{P}\left(\frac{1}{n}\sup_{0\leq x\leq 1}\sum_{j=1}^{\lfloor xn\rfloor}|V_{j}^{(n)}-V_{j}|
\geq\epsilon\right)
\nonumber
\\
&\leq
\mathbb{P}\left(\frac{1}{n}\sum_{j=1}^{n}|V_{j}^{(n)}-V_{j}|
\geq\epsilon\right)
\nonumber
\\
&=\mathbb{P}\left(\sum_{j=1}^{n}\left|n\frac{1}{\sqrt{2\beta}}\log\left|1+\frac{\sqrt{2\beta}}{n}V_{j}\right|-V_{j}\right|
\geq n\epsilon\right)
\nonumber
\\
&=\mathbb{P}\left(\left\{\sum_{j=1}^{n}\left|n\frac{1}{\sqrt{2\beta}}\log\left|1+\frac{\sqrt{2\beta}}{n}V_{j}\right|-V_{j}\right|
\geq n\epsilon\right\}\bigcap\left\{\left|\frac{\sqrt{2\beta}}{n}V_{j}\right|<\frac{1}{2},1\leq j\leq n\right\}\right)
\nonumber
\\
&\qquad
+\mathbb{P}\left(\left\{\sum_{j=1}^{n}\left|n\frac{1}{\sqrt{2\beta}}\log\left|1+\frac{\sqrt{2\beta}}{n}V_{j}\right|-V_{j}\right|
\geq n\epsilon\right\}\bigcap\left\{\left|\frac{\sqrt{2\beta}}{n}V_{j}\right|<\frac{1}{2},1\leq j\leq n\right\}^{c}\right)
\nonumber
\\
&\leq
\mathbb{P}\left(\left\{\sum_{j=1}^{n}\left|n\frac{1}{\sqrt{2\beta}}\log\left|1+\frac{\sqrt{2\beta}}{n}V_{j}\right|-V_{j}\right|
\geq n\epsilon\right\}\bigcap\left\{\left|\frac{\sqrt{2\beta}}{n}V_{j}\right|<\frac{1}{2},1\leq j\leq n\right\}\right)
\nonumber
\\
&\qquad
+\mathbb{P}\left(\bigcup_{1\leq j\leq n}\left\{\left|\frac{\sqrt{2\beta}}{n}V_{j}\right|\geq\frac{1}{2}\right\}\right).
\nonumber
\end{align}
On the one hand, since $V_{j}$ are i.i.d. $N(0,1)$,
\begin{align}\label{SuperExpI}
&\limsup_{n\rightarrow\infty}\frac{1}{n}\log\mathbb{P}\left(\bigcup_{1\leq j\leq n}\left\{\left|\frac{\sqrt{2\beta}}{n}V_{j}\right|\geq\frac{1}{2}\right\}\right)
\\
&\leq\limsup_{n\rightarrow\infty}\frac{1}{n}\log n\mathbb{P}\left(\left|V_{1}\right|\geq\frac{n}{2\sqrt{2\beta}}\right)
\nonumber
\\
&=-\infty.
\nonumber
\end{align}
On the other hand, for any $|x|<\frac{1}{2}$,
\begin{equation}
\log|1+x|=\log(1+x)=\sum_{k=1}^{\infty}(-1)^{k+1}\frac{x^{k}}{k},
\end{equation}
and
\begin{equation}
\left|\log|1+x|-x\right|\leq\sum_{k=2}^{\infty}\frac{|x|^{k}}{k}\leq x^{2}\sum_{k=2}^{\infty}\frac{1}{2^{k-2}}=2x^{2}.
\end{equation}
Therefore,
\begin{align}
&\mathbb{P}\left(\left\{\sum_{j=1}^{n}\left|n\frac{1}{\sqrt{2\beta}}\log\left|1+\frac{\sqrt{2\beta}}{n}V_{j}\right|-V_{j}\right|
\geq n\epsilon\right\}\bigcap\left\{\left|\frac{\sqrt{2\beta}}{n}V_{j}\right|<\frac{1}{2},1\leq j\leq n\right\}\right)
\\
&\leq\mathbb{P}\left(\sum_{j=1}^{n}\frac{n}{\sqrt{2\beta}}2\left(\frac{\sqrt{2\beta}}{n}V_{j}\right)^{2}\geq n\epsilon\right)
\nonumber
\\
&=\mathbb{P}\left(\sum_{j=1}^{n}V_{j}^{2}\geq\frac{n^{2}\epsilon}{2\sqrt{2\beta}}\right)
\nonumber
\\
&\leq\mathbb{E}\left[e^{\frac{1}{2}V_{1}^{2}}\right]^{n}e^{-\frac{n^{2}\epsilon}{4\sqrt{2\beta}}},
\nonumber
\end{align}
where we used the Chebyshev inequality and the fact that $V_{j}$ are i.i.d. $N(0,1)$ and $\mathbb{E}\left[e^{\frac{1}{2}V_{1}^{2}}\right]<\infty$.
Hence, we have
\begin{align}\label{SuperExpII}
&\limsup_{n\rightarrow\infty}\frac{1}{n}\log\mathbb{P}\bigg(\left\{\sum_{j=1}^{n}\left|n\frac{1}{\sqrt{2\beta}}\log\left|1+\frac{\sqrt{2\beta}}{n}V_{j}\right|-V_{j}\right|
\geq n\epsilon\right\}
\\
&\qquad\qquad\qquad\qquad
\bigcap\left\{\left|\frac{\sqrt{2\beta}}{n}V_{j}\right|<\frac{1}{2},1\leq j\leq n\right\}\bigg)=-\infty.
\nonumber
\end{align}
By \eqref{SuperExpI} and \eqref{SuperExpII}, we get the desired result.
\end{proof}

\subsection{Proofs of the Results in Section \ref{EulerEulerSection}}

\begin{proof}[Proof of Theorem \ref{EulerEulerThm}]
Notice that $0\leq Y_{k}\leq q$ and by the definition of $V_{j}^{(n)}$, it is easy to check that
\begin{equation}
|V_{j}^{(n)}|=n\frac{1}{\sqrt{2\beta}}\left|\log\left|1+\frac{\sqrt{2\beta}}{n}V_{j}\right|\right|
\leq n\frac{1}{\sqrt{2\beta}}\log\left(1+\frac{2\beta}{n}|V_{j}|\right)
\leq |V_{j}|.
\end{equation}
Therefore,
\begin{align}
&\mathbb{E}\left[\prod_{k=0}^{n-1}e^{\log\rho Y_{k}+(\frac{\sqrt{2\beta}}{n}\sum_{j=1}^{k}V_{j}^{(n)})Y_{k}}\right]
\\
&=\mathbb{E}\left[\prod_{k=0}^{n-1}e^{\log\rho Y_{k}+(\frac{\sqrt{2\beta}}{n}\sum_{j=1}^{k}V_{j}^{(n)})Y_{k}}
1_{\sup_{0\leq x\leq 1}\left|\frac{1}{n}\sum_{j=1}^{\lfloor xn\rfloor}V_{j}^{(n)}-\frac{1}{n}\sum_{j=1}^{\lfloor xn\rfloor}V_{j}\right|
<\epsilon}\right]
\nonumber
\\
&\qquad
+\mathbb{E}\left[\prod_{k=0}^{n-1}e^{\log\rho Y_{k}+(\frac{\sqrt{2\beta}}{n}\sum_{j=1}^{k}V_{j}^{(n)})Y_{k}}
1_{\sup_{0\leq x\leq 1}\left|\frac{1}{n}\sum_{j=1}^{\lfloor xn\rfloor}V_{j}^{(n)}-\frac{1}{n}\sum_{j=1}^{\lfloor xn\rfloor}V_{j}\right|
\geq\epsilon}\right]
\nonumber
\\
&\leq
\mathbb{E}\left[\prod_{k=0}^{n-1}e^{\log\rho Y_{k}+(\frac{\sqrt{2\beta}}{n}\sum_{j=1}^{k}V_{j})Y_{k}}\right]e^{n\sqrt{2\beta}q\epsilon}
\nonumber
\\
&\qquad
+\mathbb{E}\left[e^{n(\log\rho)q+\sqrt{2\beta}q\sum_{j=1}^{n}|V_{j}|}
1_{\sup_{0\leq x\leq 1}\left|\frac{1}{n}\sum_{j=1}^{\lfloor xn\rfloor}V_{j}^{(n)}-\frac{1}{n}\sum_{j=1}^{\lfloor xn\rfloor}V_{j}\right|
\geq\epsilon}\right].
\nonumber
\end{align}
By Cauchy-Schwarz inequality, 
\begin{align}
&\mathbb{E}\left[e^{n(\log\rho)q+\sqrt{2\beta}q\sum_{j=1}^{n}|V_{j}|}
1_{\sup_{0\leq x\leq 1}\left|\frac{1}{n}\sum_{j=1}^{\lfloor xn\rfloor}V_{j}^{(n)}-\frac{1}{n}\sum_{j=1}^{\lfloor xn\rfloor}V_{j}\right|
\geq\epsilon}\right]
\\
&\leq\mathbb{E}\left[e^{2[n(\log\rho)q+\sqrt{2\beta}q\sum_{j=1}^{n}|V_{j}|]}\right]^{1/2}
\mathbb{P}\left(\sup_{0\leq x\leq 1}\left|\frac{1}{n}\sum_{j=1}^{\lfloor xn\rfloor}V_{j}^{(n)}-\frac{1}{n}\sum_{j=1}^{\lfloor xn\rfloor}V_{j}\right|
\geq\epsilon\right)^{1/2}
\nonumber
\\
&=e^{n(\log\rho)q}\mathbb{E}\left[e^{\sqrt{2\beta}q|V_{1}|}\right]^{n/2}
\mathbb{P}\left(\sup_{0\leq x\leq 1}\left|\frac{1}{n}\sum_{j=1}^{\lfloor xn\rfloor}V_{j}^{(n)}-\frac{1}{n}\sum_{j=1}^{\lfloor xn\rfloor}V_{j}\right|
\geq\epsilon\right)^{1/2}.
\nonumber
\end{align}
By Lemma \ref{EquivLemma}, we get
\begin{equation}
\limsup_{n\rightarrow\infty}\frac{1}{n}\log\mathbb{E}\left[e^{n(\log\rho)q+\sqrt{2\beta}q\sum_{j=1}^{n}|V_{j}|}
1_{\sup_{0\leq x\leq 1}\left|\frac{1}{n}\sum_{j=1}^{\lfloor xn\rfloor}V_{j}^{(n)}-\frac{1}{n}\sum_{j=1}^{\lfloor xn\rfloor}V_{j}\right|
\geq\epsilon}\right]
=-\infty.
\end{equation}
Since it holds for any $\epsilon>0$, we conclude that
\begin{align}
&\limsup_{n\rightarrow\infty}\frac{1}{n}\log
\mathbb{E}\left[\prod_{k=0}^{n-1}e^{\log\rho Y_{k}+(\frac{\sqrt{2\beta}}{n}\sum_{j=1}^{k}V_{j}^{(n)})Y_{k}}\right]
\\
&\leq\limsup_{n\rightarrow\infty}\frac{1}{n}\log
\mathbb{E}\left[\prod_{k=0}^{n-1}e^{\log\rho Y_{k}+(\frac{\sqrt{2\beta}}{n}\sum_{j=1}^{k}V_{j})Y_{k}}\right].
\nonumber
\end{align}
Similarly, we can prove that
\begin{align}
&\liminf_{n\rightarrow\infty}\frac{1}{n}\log
\mathbb{E}\left[\prod_{k=0}^{n-1}e^{\log\rho Y_{k}+(\frac{\sqrt{2\beta}}{n}\sum_{j=1}^{k}V_{j}^{(n)})Y_{k}}\right]
\\
&\geq\liminf_{n\rightarrow\infty}\frac{1}{n}\log
\mathbb{E}\left[\prod_{k=0}^{n-1}e^{\log\rho Y_{k}+(\frac{\sqrt{2\beta}}{n}\sum_{j=1}^{k}V_{j})Y_{k}}\right].
\nonumber
\end{align}
Hence, from the proof for the Euler-log-Euler scheme, we proved the desired result.
\end{proof}

\subsection{Moments of $S_n$ in the Euler-Log-Euler Scheme}

We present here a method for exact computation of the positive integer 
moments of the asset price in the Euler-Log-Euler discretization, which 
were used for the numerical comparisons in Section 7.
This is a simple modification of the 
recursion relation presented in Appendix 1 of \cite{RMP}. 

Denote the moment as
\begin{equation}
M_n^{(q)} = \mathbb{E}[(S_n)^q]\,.
\end{equation}
This can be written as follows by evaluating the expectations over $\varepsilon_k$ as
\begin{eqnarray}
M_n^{(q)} &=& S_0^q \mathbb{E}\left[\prod_{k=0}^{n-1} \left(1 + \sigma_0 e^{\omega Z_k - \frac12 \omega^2 t_k} \varepsilon_k \sqrt{\tau} \right)^q\right] \\
&=&
S_0^q \mathbb{E}\left[ \prod_{k=0}^{n-1}\sum_{m=0}^{[q/2]}
\sigma_0^{2m} \tau^m \frac{q!}{(2m)!(q-2m)!} e^{2m \omega Z_k - m\omega^2 t_k}\right]\,.\nonumber
\end{eqnarray}

The moment $M_n^{(q)}$ is given by the following result.
\begin{proposition}\label{prop:Mn}
\begin{equation}\label{momsol}
M_n^{(q)} = S_0^q c_q \sum_{j=0}^{[q/2](n-1)} b_j^{0,q} \sigma_0^{2j} \tau^j ,
\end{equation}
where $c_q$ is defined in (\ref{cqdef}) and the coefficients $b_j^{i,q}$ are 
found by solving the backwards recursion
\begin{equation}\label{sol}
b_j^{i,q} = b_j^{i+1,q} + \sum_{m=1}^{[q/2]} 
b_{j-2m}^{i+1,q} \frac{(2m)!}{q! (2m-q)!} e^{2m(j-m-\frac12 ) \omega^2 t_{i+1}}
\end{equation}
with initial condition
\begin{equation}
b_0^{n-1,q} = 1\,, \quad 
b_j^{n-1,q} = 0\,, j \geq 1\,.
\end{equation}
\end{proposition}

\begin{proof}
Define the conditional expectation, conditioning on $Z_i$
\begin{equation}
\beta_i^{(q)}(Z_i) = \mathbb{E}\left[\prod_{k=i+1}^{n-1} \left(
\sum_{m=0}^{[q/2]} \sigma_0^{2m} \tau^m \frac{q!}{(2m)!(q-2m)!} e^{2m\omega Z_k - m \omega^2 t_k} \right) | Z_i \right].
\end{equation}

The moment $M_n^{(q)}$ is expressed in terms of this quantity as 
$M_n^{(q)} = S_0^q c_q \beta_0^q(0)$, with
\begin{equation}\label{cqdef}
c_q = \sum_{m=0}^{[q/2]} \frac{q!}{(2m)!(q-2m)!} \sigma_0^{2m} \tau^m\,.
\end{equation}
By the tower property of conditional expectations we note that the 
$\beta_i^{(q)}$ satisfy the recursion relation
\begin{equation}\label{rec}
\beta_i^{(q)}(Z_i) = \mathbb{E}\left[ \left(
\sum_{m=0}^{[q/2]} \sigma_0^{2m} \tau^m \frac{q!}{(2m)!(q-2m)!} 
e^{2m\omega Z_{i+1} - m \omega^2 t_{i+1}} \right)\beta_{i+1}^{(q)}(Z_{i+1}) | Z_i \right]
\end{equation}
with initial condition $\beta_{n-1}^{(q)}(Z_{n-1})=1$ at $i=n-1$. 
The solution of this recursion relation has the form
\begin{equation}\label{gen}
\beta_i^{(q)}(Z_i) = \sum_{k=0}^{2[q/2](n-i-1)}
b_k^{i,q} e^{\kappa \omega Z_i - \frac12 k^2 \omega^2 t_i} .
\end{equation}

The result (\ref{gen}) can be proved by induction in $i$. First, we note that 
it holds for $i=n-2$, as we have by explicit calculation
\begin{eqnarray}
\beta_{n-2}^{(q)}(Z_{n-2}) &=& 
\mathbb{E}\left[
\sum_{m=0}^{[q/2]} \sigma_0^{2m} \tau^m \frac{q!}{(2m)!(q-2m)!} e^{2m\omega Z_{n-1} - m \omega^2 t_{n-1}}  | Z_{n-2} \right] \\
&=& \sum_{m=0}^{[q/2]} \sigma_0^{2m} \tau^m \frac{q!}{(2m)!(q-2m)!} e^{m(2m-1)}
e^{2m\omega Z_{n-2} - m \omega^2 t_{n-2}} .
\nonumber
\end{eqnarray}

Second, substituting the expression (\ref{gen}), assumed to hold for 
$\beta_{i+1}^{(q)}(Z_{i+1})$, into the recursion (\ref{rec}), one finds 
that the same form (\ref{gen}) holds also for 
$\beta_{i+1}^{(q)}(Z_{i+1})$ with coefficients given by (\ref{sol}). 
This proves the validity of the result (\ref{gen}), and the recursion 
relation (\ref{rec}) for the coefficients $b_k^{i,q}$. Taking $i=0$ in the 
expectation $\beta_i^{(q)}(Z_i)$ reproduces the moment of the asset price 
$S_n$ as given by (\ref{momsol}). This concludes the proof of these results.

\end{proof}


\section*{Acknowledgements}

The authors are grateful to the Editor and two anonymous referees
for their helpful suggestions that greatly improved the quality of the paper.
Lingjiong Zhu is partially supported by NSF Grant DMS-1613164.
We are grateful for the suggestions and comments from the participants at 
the MCFAM Seminar at University of Minnesota, 
AMS Sectional Meeting at University of Georgia, Bachelier Congress New York, 
IAQF/Thalesians Seminar in New York.
We would like to thank Peter Carr, Paul Glasserman, Camelia Pop, Steven
Shreve and Songyun Xu for 
useful discussions.


\begin{thebibliography}{99}


\bibitem{AP2007}
L.~Andersen and V.~Piterbarg,
Moment explosions in stochastic volatility models,
Finance and Stochastics
{\bf 11}, 29-50 (2009).

\bibitem{AZ}
D.~Aristoff and L.~Zhu, On the phase transition curve in a directed 
exponential random graph model, arXiv:1404.6514[math.PR] 2014.

\bibitem{BT1}
V.~Bally and D.~Talay, The law of the Euler scheme for stochastic differential
equations: I. Convergence rate of the distribution function, 
{\em Probability Theory and Related Fields} {\bf 104}, 43-60 (1995).

\bibitem{BCM}
C.~Bernard, Z.~Cui and D.~McLeish, 
On the martingale property in stochastic 
volatility models based on time-homogeneous diffusions, 
\textit{Math. Finance} \textbf{27}, 194-223 (2017).

\bibitem{CD}
S.~Chatterjee and P.~Diaconis, Estimating and understanding exponential 
random graph models, {\em Annals of Statistics} 41, 2428-2461 (2013).

\bibitem{CS89}
M.~Chesney and L.~Scott,
Pricing European currency options: A comparison of the modified
Black-Scholes model and a random variance model, 
{\em Journal of Financial and Quantitative Analysis} {\bf 24}, 267-284 (1989).

\bibitem{Dembo} 
A. Dembo and O. Zeitouni. \emph{Large Deviations Techniques and Applications}. 2nd Edition, Springer, New York, 1998.




\bibitem{Ellis}
R.~Ellis, Entropy, Large Deviations and Statistical Mechanics, (Classics in
Mathematics), Springer, New York, 2005.

\bibitem{Forde1} 
M.~Forde and A.~Pogudin, 
The large-maturity smile for the SABR and CEV-Heston models, 
{\em Int.~J.~Th.~Appl.Finance} {\bf 16}(8) (2013).


\bibitem{review}
P.~Friz and M.~Keller-Ressel, 
Moment explosions in stochastic volatility models, 
in the {\em Encyclopedia of Quantitative Finance},
Rama Cont (Ed.), John Wiley and Sons, New York, 2013.

\bibitem{GHS}
P.~Glasserman, P.~Heidelberger and P.~Shahabuddin,
Asymptotically optimal importance sampling and stratification for
pricing path-dependent options,
{\em Math.~Finance} {\bf 9}(2), 117-152 (1999).


\bibitem{GK}
P.~Glasserman and K.-K.~Kim,
Moment explosions and stationary
distributions in affine diffusion models, 
{\em Math.~Finance} {\bf 20}(1), 1-33 (2010).

\bibitem{PGbook}
P.~Glasserman,
Monte Carlo methods in financial engineering, 
Springer, New York, 2010.


\bibitem{GS}
A.~Gulisashvili and E.~M.~Stein,
Implied volatility in the Hull-White model, 
{\em Math.~Finance} {\bf 19}(2) 303-327 (2009).

\bibitem{Guyon}
J.~Guyon,
Euler scheme and tempered distributions,
{\em Stochastic Processes and their Applications} {\bf 116}(6),
877-904 (2006).

\bibitem{SABR}
P.~Hagan, D.~Kumar, A.~Lesniewski and D.~Woodward,
Managing Smile Risk, {\em Wilmott Magazine}, pp.84-108 (2003).

\bibitem{HW}
J.~Hull and A.~White, 
Pricing of Options on Assets with
Stochastic Volatilities, {\em J.~Finance} {\bf 42}, 281-300 (1987).

\bibitem{JK}
P.~J\"ackel and C.~Kahl,
Hyp Hyp Hooray,
Wilmott Magazine 70-81, March 2008.

\bibitem{Jourdain} 
B.~Jourdain,
Loss of martingality in asset price models with
lognormal stochastic volatility, preprint 2004.

\bibitem{Jourdain2}
B.~Jourdain and M.~Sbai, 
High order discretization schemes for stochastic volatility models,
{\em J.~Comp.~Finance} {\bf 17}(2) (2013).



\bibitem{KlPl}
P.~E.~Kloeden and E.~Platen,
Numerical Solution of Stochastic Differential Equations, 
Springer, Berlin, 1992.


\bibitem{Lewis}
A.~Lewis,
Option valuation under stochastic volatility: 
with Mathematica code, 
Finance Press, Newport Beach, 2000.

\bibitem{LionsMusiela}
Pierre-Louis Lions and M.~Musiela,
Correlations and bounds for stochastic volatility models,
{\em Annales de l'Institut Henri Poincar\'e}, {\bf 24}, 1-16 (2007).

\bibitem{RMP}
D.~Pirjol,
Emergence of heavy-tailed distributions in a random multiplicative model
driven by a Gaussian stochastic process,
{\em J.~Stat.~Phys.} 154, 781-806 (2014).

\bibitem{LD}
D.~Pirjol and L.~Zhu.
On the growth rate of a linear stochastic recursion with Markovian dependence,
{\em J.~Stat.~Phys.} 160, 1354-1388 (2015).

\bibitem{PZAsian}
D.~Pirjol and L.~Zhu.
Asymptotics for the discrete-time average of the geometric Brownian motion
and Asian options. 
{\em Adv.~Appl.~Prob.} {\bf 49}, 446-480 (2017).

\bibitem{RY}
C.~Radin and M.~Yin, Phase transitions in exponential random graphs,
{\em Annals of Applied Probability} {\bf 23}, 2458-2471 (2013).


\bibitem{Sco87}
L.~Scott,
Option pricing when the variance changes randomly: 
Theory, estimation and an application,
{\em Journal of Financial and Quantitative Analysis} {\bf 22}, 419-438 (1987).

\bibitem{Sin}
C.A.~Sin, 
Complications with stochastic volatility models, 
\textit{Adv.~Appl.~Probab.} \textbf{30}(1), 256-268 (1998).

\bibitem{TT}
D.~Talay and L.~Tubaro,
Expansion of the global error for numerical schemes solving stochastic
differential equations, 
{\em Stochastic Analysis and Applications}, {\bf 8}(4), 483-509 (1990).


\bibitem{VaradhanII} S. R. S. Varadhan. 
\textit{Large Deviations and Applications}, SIAM, Philadelphia, 1984.

\bibitem{WPW}
T.~H.~Wang, P.~Laurence and S.~L.~Wang,
Generalized uncorrelated SABR models with a high degree of symmetry,
{\em Quantitative Finance}, 1-17 (2010).



\end{thebibliography}
\end{document}